\documentclass[11pt,a4paper]{article}
\usepackage{a4wide}         
\usepackage{amsmath,amsfonts,amssymb,amstext,amsthm}
\usepackage{filecontents}
\usepackage{bbm}  
\usepackage{tikz}
\usetikzlibrary{decorations.pathmorphing}
\usetikzlibrary{decorations.pathreplacing}

\newcommand{\Nesetril}{Ne\v set\v ril}
\newcommand{\numP}{\ensuremath{\#\mathrm{P}}} 
\newcommand{\HomsTo}[1]{\ensuremath{\oplus \prb{HomsTo}{#1}}}
 
\newcommand{\parp}{\ensuremath{\oplus\mathrm{P}}}
\newcommand{\FP}{\ensuremath{\mathrm{FP}}}
\newcommand{\prb}[1]{\textsc{#1}}
\newcommand{\parhcol}[1]{\ensuremath{\oplus \prb{HomsTo}{#1}}}
\newcommand{\pinnedparhcol}{\ensuremath{\oplus \prb{PinnedHomsTo}H}}
\newcommand{\rpinnedparhcol}{\ensuremath{\oplus r\prb{-PinnedHomsTo}H}}
\newcommand{\paris}{\ensuremath{\oplus\prb{IS}}}

\newcommand{\calI}{\ensuremath{\mathcal{I}}}

\newcommand{\posints}{\ensuremath{\mathbb{N}_{>0}}}
\newcommand{\dist}{\ensuremath{d}}
\newcommand{\indicator}{\mathbbm{1}}
\newcommand{\GFtwo}{\ensuremath{\mathrm{GF}(2)}}
\DeclareMathOperator{\Hom}{Hom}
\DeclareMathOperator{\HomPin}{HomPin}
\DeclareMathOperator{\Aut}{Aut}
\DeclareMathOperator{\Orb}{Orb}
\newcommand{\bj}{\mathbf{j}}
\newcommand{\bv}{\mathbf{v}}
  \newcommand\calJ{\mathcal J}

  \newcommand\pin{p} 
 \newcommand{\id}{\mathrm{id}}

\newcommand{\ceil}[1]{\lceil {#1}\rceil}
\newcommand{\floor}[1]{\lfloor {#1}\rfloor}

\newtheorem{theorem}{Theorem} 
\newtheorem{lemma}[theorem]{Lemma}
\newtheorem{corollary}[theorem]{Corollary}
\newtheorem{conjecture}[theorem]{Conjecture}
\theoremstyle{definition}
\newtheorem{definition}[theorem]{Definition}
\newtheorem{observation}[theorem]{Observation}

\makeatletter
\def\imod#1{\allowbreak\mkern10mu({\operator@font mod}\,\,#1)}
\makeatother

\date{25 April, 2014}                                        
\title{The Complexity of Counting\\ Homomorphisms to Cactus Graphs Modulo~2%
\thanks{
The research leading to these results has received funding from 
the European Research Council under the European Union's Seventh Framework Programme (FP7/2007-2013) ERC grant agreement no.\ 334828. The paper 
reflects only the authors' views and not the views of the ERC or the European Commission. The European Union is not liable for any use that may be made of the information contained therein.
Some of the initial research was supported by the EPSRC grant EP/I011528/1.
An extended abstract of this paper  appeared in the proceedings of STACS 2014.
 Authors' address: Department of Computer Science, University of Oxford, Wolfson Building, Parks Road, Oxford, OX1~3QD, UK.}}
\author{Andreas G\"obel \and Leslie Ann Goldberg \and David Richerby}

\begin{document}


\maketitle 

\begin{abstract}
A homomorphism from a graph~$G$ to a graph~$H$ 
is a function from~$V(G)$ to~$V(H)$ that preserves edges.
Many combinatorial structures that
arise in mathematics and in computer science
can be represented naturally as graph homomorphisms and as weighted sums of
graph homomorphisms.
In this paper, we study the
complexity of counting homomorphisms modulo~$2$. 
The complexity of modular counting  
was introduced by Papadimitriou and Zachos and it has been pioneered by
Valiant who famously introduced a problem for which counting modulo~$7$
is easy where counting modulo~$2$ is intractable. Modular counting
provides a rich setting in which to study the structure of homomorphism 
problems. In this case,
the structure of the graph~$H$
has a big influence on the complexity of the problem. Thus,
our approach is
graph-theoretic. We give a complete solution for the class of cactus graphs,
which are connected graphs in which every edge belongs to at most one cycle.
Cactus graphs arise in many applications such as the modelling of wireless sensor networks
and the comparison of genomes. We show that, for some cactus graphs~$H$,
counting homomorphisms to~$H$ modulo~$2$
can be done in polynomial time. For every other fixed cactus graph~$H$, the problem
is complete in the complexity class~$\parp$ which is a wide complexity class
to which every problem in the polynomial hierarchy can be reduced (using randomised reductions).
Determining which $H$ lead to tractable problems can be done in polynomial time.
Our result builds upon the work of Faben and Jerrum, who gave a dichotomy for the case in which $H$ is a tree.
\end{abstract}


\section{Introduction}

A homomorphism from a graph~$G$ to a graph~$H$ 
is a function from~$V(G)$ to~$V(H)$ that preserves edges. That is,
the function maps every edge of~$G$ to an edge of~$H$.
Many combinatorial structures that
arise in mathematics and in computer science
can be represented naturally as graph homomorphisms. 
For example, the proper $q$-colourings of a graph~$G$
correspond to the homomorphisms from~$G$ 
to the $q$-clique. The independent sets of~$G$
correspond to the homomorphisms from~$G$
the $2$-vertex connected graph with one self-loop
(vertices of~$G$ which are mapped to the self-loop are out
of the corresponding independent set; vertices which are
mapped to the other vertex are in it).
Partition functions in statistical physics 
such as the Ising model, the Potts model, and the
hard-core model arise naturally as weighted
sums of homomorphisms. See, for example, \cite{BG, GGJT}.

In this paper, we study the
complexity of counting homomorphisms modulo~$2$. 
For graphs $G$ and $H$, let $\Hom(G,H)$ denote 
the set of homomorphisms from~$G$ to~$H$. 
For each fixed~$H$,
we study the computational problem~$\HomsTo{H}$,
which is the problem of computing~${\left|\Hom(G,H)\right|} \bmod 2$,
given an input graph~$G$.
 
The structure of the graph~$H$
has a big influence on the complexity of~$\HomsTo{H}$. 
For example, consider the graphs~$H_1$, $H_2$ and~$H_3$
depicted in Figure~\ref{fig:hard-easy-hard}.
Our result implies that $\HomsTo{H_1}$ is 
complete for the class $\parp$ 
(with respect to polynomial-time Turing reductions).
 The graph~$H_2$ is   constructed by moving the top right ``bristle''
from~$H_1$ down to the bottom right.
Under the standard assumption
that $\parp \neq \FP$ (which will be explained below),
moving this bristle
makes the problem easier --- our result implies that $\HomsTo{H_2}$ is  
solvable in polynomial time.
$H_3$ is constructed by moving the top bristle
from left to right in~$H_2$. This again makes the problem more difficult ---
$\HomsTo{H_3}$ is $\parp$-compete.
\begin{figure}[t]
\begin{center}
\begin{tikzpicture}[scale=.4,node distance = 1.5cm]
\tikzstyle{vertex}=[fill=black, draw=black, circle, inner sep=1.5pt]

\foreach \x in {4,6} \node[vertex] (a\x) at (\x,6) {};
\foreach \x in {2,4,6} \node[vertex] (b\x) at (\x,4) {};
\foreach \x in {0,2,4,6} \node[vertex] (c\x) at (\x,2) {};
\foreach \x in {0,2,4,6} \node[vertex] (d\x) at (\x,0) {};

\draw (a4)--(a6)--(b6)--(b4)--(c4)--(c6)--(d6)--(d4)--(c4)--(c2)--(d2)--(d0)--(c0)--(c2)--(b2)--(b4)--(a4);

\node[vertex] (1) at (-1,3) {};
\node[vertex] (2) at (3,-1) {};
\node[vertex] (4) at (3,7) {};
\node[vertex] (3) at (7,7) {};

\draw (c0)--(1);
\draw (d4)--(2);
\draw (a6)--(3);
\draw (a4)--(4);

\node at (3,-3) {$H_1$};
\end{tikzpicture}
\hspace{1.2cm}
\begin{tikzpicture}[scale=.4,node distance = 1.5cm]
\tikzstyle{vertex}=[fill=black, draw=black, circle, inner sep=1.5pt]

\foreach \x in {4,6} \node[vertex] (a\x) at (\x,6) {};
\foreach \x in {2,4,6} \node[vertex] (b\x) at (\x,4) {};
\foreach \x in {0,2,4,6} \node[vertex] (c\x) at (\x,2) {};
\foreach \x in {0,2,4,6} \node[vertex] (d\x) at (\x,0) {};

\draw (a4)--(a6)--(b6)--(b4)--(c4)--(c6)--(d6)--(d4)--(c4)--(c2)--(d2)--(d0)--(c0)--(c2)--(b2)--(b4)--(a4);

\node[vertex] (1) at (-1,3) {};
\node[vertex] (2) at (3,-1) {};
\node[vertex] (3) at (7,-1) {};
\node[vertex] (4) at (3,7) {};

\draw (c0)--(1);
\draw (d4)--(2);
\draw (d6)--(3);
\draw (a4)--(4);

\node at (3,-3) {$H_2$};
\end{tikzpicture}
\hspace{1.2cm}
\begin{tikzpicture}[scale=.4,node distance = 1.5cm]
        \tikzstyle{vertex}=[fill=black, draw=black, circle, inner sep=1.5pt]

\foreach \x in {4,6} \node[vertex] (a\x) at (\x,6) {};
\foreach \x in {2,4,6} \node[vertex] (b\x) at (\x,4) {};
\foreach \x in {0,2,4,6} \node[vertex] (c\x) at (\x,2) {};
\foreach \x in {0,2,4,6} \node[vertex] (d\x) at (\x,0) {};

\draw (a4)--(a6)--(b6)--(b4)--(c4)--(c6)--(d6)--(d4)--(c4)--(c2)--(d2)--(d0)--(c0)--(c2)--(b2)--(b4)--(a4);

\node[vertex] (1) at (-1,3) {};
\node[vertex] (2) at (3,-1) {};
\node[vertex] (3) at (7,-1) {};
\node[vertex] (4) at (7,7) {};

\draw (c0)--(1);
\draw (d4)--(2);
\draw (d6)--(3);
\draw (a6)--(4);

\node at (3,-3) {$H_3$};
\end{tikzpicture}
\caption{$\HomsTo{H_1}$ and $\HomsTo{H_3}$ are $\parp$-complete, but $\HomsTo{H_2}$ is in $\FP$.\label{fig:hard-easy-hard}}
\end{center}
\end{figure}
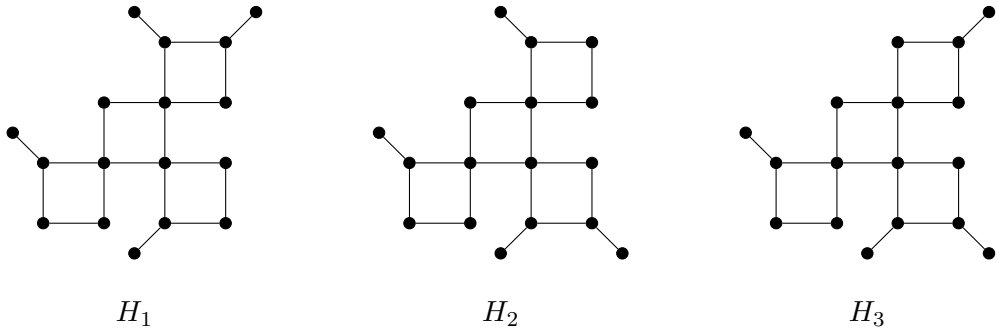
The goal of this research is to study the complexity of
$\HomsTo{H}$ for every fixed graph~$H$ and
to determine for which graphs~$H$ the problem is
in~$\FP$, for which graphs it is $\parp$-complete, and whether
there are any graphs~$H$ for which the problem has intermediate complexity.
In this paper, we give a complete solution to this problem for 
the class of \emph{cactus graphs}.

A cactus graph is a connected graph  
in which every edge belongs to at most one cycle.
Cactus graphs were first defined by Harary and Uhlenbeck~\cite{HU}
who attributed them to the physicist Husimi and therefore called them
\emph{Husimi Trees}. Cactus graphs  
arise, for example, in the modelling of 
wireless sensor networks~\cite{wireless}
and in the comparison of genomes~\cite{genome}.
Some NP-hard graph problems
can be solved in polynomial time on cactus graphs~\cite{cactusalg}.

\subsection{The complexity of modular counting}

The complexity of modular counting is an interesting topic
with some surprising results and
we only mention a few highlights here.
First, it is important to understand that~$\parp$ 
(which was first studied in~\cite{GP86, PZ83})
is a very large complexity class. 
We will think of this class from the 
point of view of function computation, so~$\parp$
consists of all problems of the form 
``compute $f(x) \bmod 2$'' where computing~$f(x)$
is a problem in~$\numP$.
$\parp$ is sufficiently powerful that there is randomised
polynomial-time reduction \cite{T91} 
from every problem in the polynomial hierarchy to some problem in~$\parp$.
Thus, subject to the natural hypothesis that
problems in the higher levels of the polynomial hierarchy
are not solvable in (randomised) polynomial time, 
we conclude that $\parp$-complete problems are much harder
than problems in~$\FP$,
which
is the class of of function-computation problems that are solvable in polynomial
time.

The complexity of counting modulo~$2$ 
is different from the complexity of 
decision problems or counting problems.  
First, consider an NP-complete decision problem.
The mod-$2$ counting version of this problem can be intractable, 
as you might expect (for example, counting vertex covers
or independent sets modulo~$2$ is $\parp$-complete~\cite{Val06})
but it can also be
tractable.
As an example, consider the problem of 
counting proper $3$-colourings of a graph modulo~$2$.
There are an even number of $3$-colourings
that use exactly $3$ colours, since there are $6$ permutations
of these colours. There are also an even number of $3$-colourings
that use exactly $2$ colours, since 
the colours can be swapped.
But it is easy to count $1$-colourings. Thus, it is easy 
to count all proper colourings modulo~$2$.
Next, consider a $\numP$-complete counting problem.
The mod-$2$ counting version of this problem
can be intractable or tractable, as the examples 
given above illustrate. As another example where the
mod-$2$ counting version is tractable, consider the problem of computing the permanent
of a matrix
modulo~$2$. Since $-1\equiv 1\bmod 2$, the
permanent is equal to the determinant modulo~$2$, so it can 
readily be computed in polynomial time.

Another interesting aspect of modular counting is the fact that the
value of the modulus can affect the tractability of the
problem. As an example, consider 
the well-known work of Valiant~\cite{Val06}
which identified a certain satisfiability problem
where satisfying assignments are easy to compute modulo~7
and difficult to compute modulo~2.

\subsection{Dichotomies for graph homomorphism problems}

Determining the border between tractability and intractability for large classes of modular counting problems is an
important step towards understanding the structure of the problems themselves.
In this paper we work within the context of graph homomorphism problems
because graph homomorphisms are general enough to
capture a wide variety of combinatorial problems,  
yet they exhibit sufficient structure that dichotomies exist.
Hell and \Nesetril~\cite{HN} pioneered this direction
by completely classifying undirected graphs according to the difficulty
of the graph homomorphism decision problem.
They showed if a fixed graph~$H$
has a self-loop, or is bipartite then the problem of determining whether
an input graph has a homomorphism to~$H$ is in P.
For every other fixed graph~$H$, the decision problem is NP-complete.

Over recent years, dichotomy
theorems have also been established for the problem of counting
graph homomorphisms and computing weighted sums of homomorphisms.
Dyer and Greenhill~\cite{DG}
showed that the problem of counting homomorphisms to~$H$
is solvable in polynomial time if every component of~$H$ is
an isolated vertex,
a complete graph with all self-loops present, or
a complete bipartite graph with no self-loops. For every other~$H$,
it is~$\numP$-complete. In particular, there are no graphs~$H$
for which the problem has intermediate complexity.
This dichotomy was extended to the problem of
computing weighted sums of homomorphisms to~$H$.
A dichotomy was given by Bulatov and Grohe \cite{BG}
for the case where the weights are positive,
by Goldberg, Grohe, Jerrum and Thurley~\cite{GGJT}
for the case where the weights are real,
and by Cai, Chen and Lu~\cite{CL} for the case where the weights are complex.

\subsection{The complexity of counting graph homomorphisms modulo~$2$}

The first results on the complexity of counting graph homomorphisms modulo~$2$
were obtained by Faben and Jerrum~\cite{Faben,FJ} who
made some important structural discoveries which we also use.

An \emph{involution} of a graph is an automorphism of order~$2$.
If $\sigma$ is an automorphism of a graph~$H$
then $H^\sigma$ denotes the  subgraph of~$H$ induced by the fixed points of~$\sigma$.

\begin{lemma} (\cite[Lemma  3.3]{FJ})
\label{lem:same}
If $H$ is a graph and $\sigma$ is an involution of~$H$
then, for any graph $G$, 
${\left| \Hom(G,H) \right|} \equiv {\left| \Hom(G,H^\sigma) \right|} \imod 2$.
\end{lemma}

The lemma is useful because it enables us to reduce the problem of
counting homomorphisms to~$H$ modulo~$2$
to the problem of counting  homomorphisms to~$H^\sigma$ modulo~$2$. 
This leads naturally
to the idea of reduction by involutions. 
 Let $\rightarrow$ be the relation on graphs  where
$H \rightarrow H'$ if and only if there is an involution~$\sigma$ of~$H$
such that $H'=H^\sigma$. Let
$\rightarrow^*$ be the transitive closure of  $\rightarrow$.
Faben and Jerrum showed that repeatedly applying $\rightarrow$ to a graph~$H$
reduces~$H$ to a unique involution-free graph, up to isomorphism.

\begin{lemma} \label{lem:unique}
(\cite[Theorem  3.7]{FJ})
Let $H$ be any graph.
If $H\rightarrow^* H'$ and $H\rightarrow^* H''$ and $H'$ and $H''$ are involution-free
then $H'$ is isomorphic to $H''\!$.
\end{lemma}

Finally,  they showed that,
to classify the complexity of counting homomorphisms to~$H$,
one only needs to study the complexity of counting homomorphisms
to its connected components. 

\begin{lemma} (\cite[Theorem 6.1]{FJ})
\label{lem:components}
Let $H$ be an involution-free graph. If $H$ has a connected
component $H_1$
such that $\parhcol{H_1}$ is   $\parp$-hard with respect to polynomial-time
Turing reductions, then $\parhcol{H}$ is also $\parp$-hard.
\end{lemma}

To complement Lemma~\ref{lem:components}, it is easy
to see that if $\parhcol{H_i}$ is solvable in polynomial time for  
every connected component~$H_i$ of~$H$
then $\parhcol{H}$ is also solvable in polynomial time. 
Faben and Jerrum used the structural results to
give a  dichotomy for the complexity of $\parhcol{H}$
when $H$ is a tree. Define the ``involution-free reduction'' $H'$ of a graph~$H$
to be the lexicographically-minimal involution-free graph
such that $H\rightarrow^* H'$.
We can then state  their result as follows.

\begin{theorem}(\cite[Theorem  3.8]{FJ})
\label{thm:faben}
If $H$ is a tree then $\parhcol{H}$ is $\parp$-complete
if the involution-free reduction of $H$ has more than one vertex.
Otherwise, it is in~\FP.
\end{theorem}

Every involution-free tree is asymmetric, which means that
it has no non-trivial automorphisms. Thus, the technical work
of proving Theorem~\ref{thm:faben}
is to consider every asymmetric tree~$H$
with more than one vertex and show that $\parhcol{H}$ is $\parp$-hard.
Fortunately, this can be done without too much technical complexity.

Developing a dichotomy to cover all graphs 
seems to be much more difficult
and even the dichotomy for cactus graphs requires a 
substantial technical effort, as we will see.
Nevertheless,  there is a general conjecture
as to what the outcome would be.

\begin{conjecture}[Faben and Jerrum]
\label{conj:faben}
Let $H$ be a (not necessarily simple) graph.
$\parhcol{H}$ is in~\FP\ if the involution-free reduction of~$H$
is the empty graph, a singleton vertex (with or without a self-loop)
or a graph with two isolated vertices, exactly one of which has a self-loop.
Otherwise, it is $\parp$-complete.
\end{conjecture}

\subsection{Our result}

Recall that a cactus graph is a connected, simple graph
in which every edge belongs to at most one cycle.
Our main result gives a proof of Faben and Jerrum's conjecture
for cactus graphs.

\begin{theorem}
\label{thm:main}
Let $H$ be a simple graph in which every edge belongs to at most one cycle.
If the involution-free reduction of~$H$ has at 
most one vertex then $\parhcol{H}$ is solvable in polynomial time.
Otherwise, $\parhcol{H}$ is complete for $\parp$ with respect to polynomial-time Turing reductions.
\end{theorem}

To prove Theorem~\ref{thm:main} 
we have to investigate all involution-free cactus graphs, and
not just those that happen to be asymmetric.
The reason for this is that, unlike the situation for trees, there
are involution-free cactus graphs that have non-trivial automorphisms.
An example is the graph~$H_4$ in Figure~\ref{fig:easy-hard}.
This graph has an automorphism of order~$3$ which rotates the cycle.
It has no automorphism of order~$2$.
Incidentally, it is easy to see that the graph~$H_5$ in the figure
has an involution that moves all vertices. Thus,
$\parhcol{H_5}$ is in~\FP. Our result implies that $\parhcol{H_4}$ is
$\parp$-complete.
\begin{figure}[t]
\begin{center}
\begin{tikzpicture}[scale=.4,node distance = 1.5cm]
        \tikzstyle{vertex}=[fill=black, draw=black, circle, inner sep=1.5pt]

\node (c1) at (0,0) {$H_4$};
\draw (c1) circle (2);

\foreach \x in {0,40,80,120,160,200,240,280,320} \node[vertex] (\x) at (\x:2) {};
\foreach \x in {40,80,160,200,280,320} \node[vertex] (b\x) at (\x:3) {};
\foreach \x in {80,200,320} \node[vertex] (c\x) at (\x:4) {};

\foreach \x in {40,80,160,200,280,320} \draw (\x)--(b\x);
\foreach \x in {80,200,320} \draw (b\x)--(c\x);
\node at (240:4) {};
\end{tikzpicture}\hspace{1.5cm}
\begin{tikzpicture}[scale=.4,node distance = 1.5cm]
        \tikzstyle{vertex}=[fill=black, draw=black, circle, inner sep=1.5pt]
\node (c2) at (0,0) {$H_5$};
\draw (c2) circle (2);

\foreach \x in {0,30,60,90,120,150,180,210,240,270,300,330} \node[vertex] (\x) at (\x:2) {};
\foreach \x in {30,60,120,150,210,240,300,330} \node[vertex] (b\x) at (\x:3) {};
\foreach \x in {60,150,240,330} \node[vertex] (c\x) at (\x:4) {};

\foreach \x in {30,60,120,150,210,240,300,330} \draw (\x)--(b\x);
\foreach \x in {60,150,240,330} \draw (b\x)--(c\x);
\node at (80:4) {};
\end{tikzpicture}
\caption{\parhcol{H_4} is $\parp$-complete but $\parhcol{H_5}$ is in~\FP.\label{fig:easy-hard}}
\end{center}
\end{figure}

In order to prove 
the hardness result in Theorem~\ref{thm:main},  we introduce three graph-theoretic notions: hardness gadgets, partial hardness
gadgets, and mosaics. Hardness gadgets and partial hardness gadgets
are, as the name suggests, structures for proving $\parp$-hardness.
Mosaics are graphs built on unions of $4$-cycles. They are what is left in 
inductive cases where hardness gadgets don't exist and we use them inductively
to prove overall hardness.
Our approach is therefore recursive: we show how to decompose involution-free
cactus graphs at cut vertices
in such a way that every  component contains at least
one of these three induced structures.
We then show how to combine these structures to obtain hardness
gadgets in the original graph. 
If an asymmetric graph~$H$ contains a hardness gadget,
then it is relatively easy to show that  $\parhcol{H}$ is
$\parp$-compete --- the proof is by reduction from the problem
of counting weighted independent sets modulo~$2$, generalising  the argument for trees.
We will discuss the situation in which $H$ 
is not asymmetric presently.

Even when $H$ is asymmetric,
the most difficult part of the argument is showing that every non-trivial
involution-free cactus graph does actually contain a hardness
gadget. The presence of cycles in the  graph 
greatly
complicates the
structure of the argument, hence the need to 
define hardness gadgets, partial hardness gadgets and mosaics and
to 
decompose cactus graphs into components with these three
different structures, which can then be combined to form hardness gadgets.

When the graph has non-trivial automorphisms,
there is a further complication.  
Suppose that~$G$ and~$H$ are graphs and
that $p$ is a function from $V(G)$ to $2^{V(H)}$.
A homomorphism~$f$ from~$G$ to~$H$ is said to satisfy
the ``pinning'' function~$p$ if, for every $v\in V(G)$, 
we have $f(v)\in p(v)$. 
Now suppose that $H$ is an involution-free graph containing a hardness-gadget.
The high-level strategy for proving that $\HomsTo{H}$ is $\parp$-hard
is to first reduce
the problem of counting weighted independent sets
modulo~$2$ to the
problem of counting pinned homomorphisms from~$G$ to~$H$
(modulo~$2$) and then to reduce the latter problem to $\HomsTo{H}$.
This pinning approach 
has been used successfully in dichotomy theorems
in related domains~\cite{BG, CH,  wbool}.
When $H$ is asymmetric, the application
of pinning works smoothly.
Building on work of Lov\'asz~\cite{lovasz},
Faben and Jerrum reduced the pinned problem
to the unpinned one for the case in which the pinning function
pins some vertex to an orbit in the automorphism group of~$H$.
When $H$ is asymmetric (as it is when $H$ is a tree), the orbit is just
a single vertex, and this is just what is required.
If $H$ is not asymmetric then we do not know how to
pin a vertex of~$G$ to a particular vertex in~$H$.
To get around this,
we augment~$G$ with a copy of~$H$ and we pin every
vertex in the copy to its own orbit in the automorphism group of~$H$.
We show that every  homomorphism 
from  an involution-free cactus graph 
to itself
that respects the orbits of all of its vertices
is in fact an automorphism of~$H$, and this enables us to solve the problem.

Theorem~\ref{thm:main} gives a dichotomy for cactus graphs.
If the involution-free reduction of~$H$
has at most one vertex then $\parhcol{H}$ is in~\FP. Otherwise, it is
$\parp$-complete. Furthermore,
the meta-problem of determining which is the case, given input~$H$,
is  computationally
easy.  Finding an involution of~$H$ reduces
in polynomial time to computing the size of $H$'s automorphism group modulo~$2$.
The latter problem is in~\FP\ for cactus graphs because, for example, cactus graphs are planar and have tree-width at most~2.

All of the graph-theoretic work in this paper depends heavily on the structure of cactus graphs. For example,
the construction of hardness gadgets relies on the fact that removing the edges of a cycle splits the vertices of that cycle into different components. New ideas would be needed to generalise our result to a wider class of graphs.

\subsection{Related Work}
 
Modular counting has also been considered
in the context of Boolean 
constraint satisfaction problems (CSPs).
Graph homomorphism problems are CSPs but they
are not known to be representable as Boolean CSPs.
The known results are as follows. Faben  
\cite{Fab08, Faben} provided a modular counting dichotomy for 
unweighted Boolean CSPs. This was extended by
Guo et al.\ \cite{GHLX11} 
to the weighted case.
Guo et al.\ \cite{GLV13} also provided a dichotomy for the solution of
Boolean Holant problems modulo~2.


\section{Preliminaries}
   
   Our analysis is based primarily on decomposing graphs, particularly
by splitting them at cut vertices, as follows.

\begin{definition}
Given a graph $H$ with a cut vertex $v$, let $H'_1,\dots,H'_\kappa$ be the connected components of $H-\{v\}$.
The \emph{split} of $H$ at $v$ is the set of graphs $\{H_1,\dots,H_\kappa\}$, where $H_j=H[V(H'_j)\cup\{v\}]$.
 \end{definition}

Our proofs of \parp{}-completeness are based on a weighted generalisation of the independent set problem.  

\begin{definition} 
Let $\calI(G)$ denote the set of independent sets of a graph~$G$.
The \emph{Parity Generalised Independent Set} problem 
$\paris(\lambda,\mu)$
is defined as follows, where the parameters~$\lambda$ and~$\mu$ are natural numbers.\\
\emph{Name:} $\paris(\lambda,\mu)$.\\
\emph{Input:} A graph $G$.\\
\emph{Output:} 
$\displaystyle Z_{\lambda,\mu}(G) = \sum_{J\in \calI(G)} {\lambda}^{|J|} {\mu}^{|V(G)\setminus J|} \imod 2$.
\end{definition}

The following observation follows from the fact that counting independent
sets modulo~$2$ is $\parp$-complete~\cite{Val06}.
\begin{observation}
\label{obs:paris-complexity}
If $\lambda$ or~$\mu$ is even, then
$\paris(\lambda,\mu)$ is in $\FP$.
Otherwise, for every graph,
$Z_{\lambda,\mu}(G) \equiv |\calI(G)| \imod 2$, so   $\paris(\lambda,\mu)$ is $\parp$-complete.\end{observation}

We use the following notation. 
$[n]$ denotes the set $\{1,\ldots,n\}$.
Given two graphs $G$ and $H$ (not necessarily vertex-disjoint), let $G\cup H$ denote the graph
$(V(G)\cup V(H), E(G)\cup E(H))$.  
If $E$ is a set of edges, let $V(E)$ denote the set of endpoints of edges in~$E$
and let
$G \cup E$ denote the graph $G \cup (V(E),E)$.
Given sets
$V'\subseteq V(G)$ and $E'\subseteq E(G)$,
let $G-V' = G[V(G)\setminus V']$ and
let $G-E' = (V(G),E(G)\setminus E')$. We use the phrase ``$j$-walk'' in a graph to 
refer to a walk of length~$j$.

In drawings of graphs, a wavy line between two vertices indicates an
omitted portion of the graph whose structure is unimportant, except
that it contains a unique shortest path between the endpoints of the
wavy line.

Graphs are simple: they do not have multiple edges or self-loops. 
We use $\dist_H(u,v)$ to denote the length of a shortest path from~$u$ to~$v$ in~$H$. If $S\subseteq V(H)$ then $\dist_H(u,S)$ denotes 
$\min\{\dist_H(u,v) \mid v \in S\}$.  
We use $\Gamma_H(v)$ to denote the set of neighbours of vertex~$v$ in~$H$
and $\deg_H(v)$ to denote the degree of~$v$ in $H$.
 We normally view a path~$P$ as a graph with vertex set $V(P)$ and edge set $E(P)$.
However, where convenient, we specify~$P$ by simply listing the vertices of
the path, in order.  As such, if $P=x_1\dots x_\ell$ and $P'=y_1\dots
y_{\ell'}$, we write $Pz$ for the path $P\cup \{(x_\ell,z)\}$, $PP'$
for the path $P\cup \{(x_\ell,y_1)\}\cup P'$ and so on.
We also use $\ell(P)$ to denote $|E(P)|$, the length of the path~$P$.
Similarly, we view a cycle~$C$ as a graph, 
but we sometimes specify~$C$ by listing its vertices in order.
The length of a cycles is $\ell(C)=|E(C)|$.
 
For $v_1,\ldots,v_r\in V(H)$, we use the
tuple $(H,v_1,\ldots,v_r)$ to denote the graph~$H$ together with the (ordered) list of distinguished vertices.
A graph with a single distinguished vertex is often called a
\emph{rooted graph}, where the distinguished vertex is the
\emph{root}. 
$(H,v_1,\ldots,v_r)$ is said to be \emph{involution-free} if 
$H$ has no
involution
that fixes each of the distinguished vertices.
We use $\Aut(H)$ to denote the automorphism group of~$H$
and, for $v\in V(H)$, we use $\Orb_H(v)$ to denote the
set of vertices of~$H$ in the orbit of~$v$
under the action of~$\Aut(H)$.


\section{Pinning}
 
As we discussed in the introduction,
pinning is a useful technique for obtaining hardness
results related to graph homomorphisms and constraint satisfaction problems.
Where it can be achieved, pinning allows 
a reduction from the ``list'' version of the problem,  
which considers homomorphisms from~$G$ to~$H$
in which the images of the vertices of~$G$ are constrained by lists,
to the ordinary version of the problem.
The list version is referred to as the ``conservative case'' in 
the constraint satisfaction community.

As in the introduction, let $\pin$ be a function  from $V(G)$
to $2^{V(H)}$.
A homomorphism $f\in \Hom(G,H)$ satisfies
the pinning function~$\pin$
if, for every $v\in V(G)$, 
$f(v)\in \pin(v)$.   
Let $\HomPin(G,H,\pin)$ be the set of homomorphisms
from~$G$ to~$H$ that satisfy the pinning function~$\pin$.
It will be useful to consider the following computational problem, which is parameterised by
a graph~$H$.
\begin{description}
\item  \emph{Name:} $\pinnedparhcol$.
\item \emph{Input:} A graph $G$ and a pinning function $\pin\colon V(G) \to 2^{V(H)}$. 
\item \emph{Output:} ${\left|\HomPin(G,H,\pin)\right|} \imod 2$.
\end{description} 

Faben and Jerrum
give a polynomial-time Turing 
reduction from $\pinnedparhcol$ to $\parhcol{H}$
\cite[Cor 4.18]{FJ}
for the following special case: 
for vertices $v_1,v_2\in V(G)$ and $h_1,h_2\in V(H)$,
$p(v_1)=\Orb_H(h_1)$ and $p(v_2)=\Orb_H(h_2)$;
for 
every other vertex~$w$ of~$G$, $\pin(w)=V(H)$, 
so vertex~$w$ 
is unconstrained
by the pinning. 
We require a generalisation in which more vertices of~$G$ 
are constrained by the pinning.
We start by introducing the relevant concepts and notation.

Given rooted graphs $(G,x)$ and $(G'\!,x')$, the graph $(G,x)\cdot
(G'\!,x')$ is formed by taking the union of disjoint copies of $G$ and
$G'$ and identifying $x$ and~$x'\!$.
For each $y\in H$, we would like to produce a rooted graph
$(\Theta_{H,y},z)$ such that, for any rooted graph $(G,x)$, the number
of homomorphisms from~$G$ to~$H$  in which~$x$ is pinned to~$y$ is congruent
modulo~2 to the number of 
homomorphisms 
from $(G,x)\cdot
(\Theta_{H,y},z)$ to~$H$.  
However, it is not clear that such graphs always
exist so Faben and Jerrum adopts  more subtle approach.

Suppose that $(G,x)$ is   a rooted graph   and 
that~$H$ is an involution-free graph with $V(H) =
\{h_1, \dots, h_m\}$.
For $i\in[m]$,
let $\pin_i$ be the pinning function
\begin{equation*}
    \pin_i(y) = \begin{cases}
        \ \{h_i\} & \text{if $y=x$} \\
        \ V(H)    & \text{otherwise.}
    \end{cases}
\end{equation*}
Let $ v_i = {\left| \HomPin(G,H,\pin_i) \right| } \imod 2$.
Finally, let $\bv_H(G,x)$ be the vector
$(v_1,\dots, v_m)$.
We say that a vector 
$(u_1,\ldots,u_m)$
in $\GFtwo^m$ is \emph{consistent for~$H$}
if, for every pair of vertices $h_i$ and $h_j$ in
the same orbit of $\Aut(H)$, 
$u_i=u_j$.
Faben and Jerrum observe that $\bv_H(G,x)$ is consistent for~$H$.

Let $+$ and $*$ be component-wise addition and multiplication of
vectors in $\GFtwo^m\!$.   Then we have~\cite[Lemma 4.11]{FJ}  
\begin{equation*}
    \bv_H((G,x)\cdot(G'\!,x'),x) = \bv_H(G,x) * \bv_H(G'\!,x').
\end{equation*}

Now, given a graph~$H$ with $V(H) =
\{h_1, \dots, h_m\}$ 
and a vector ${\mathbf v}\in \GFtwo^m$,
we say that ${\mathbf v}$ is  
\emph{implementable for $H$} if there is a set $\{(\Theta_1,z_1),
\dots, (\Theta_k,z_k)\}$ of graphs such that ${\bf v} = \sum_{j=1}^k
\bv_H(\Theta_j,z_j)$.   The key result allowing pinning is the
following~\cite[Lemmas~4.14--4.16]{FJ}

\begin{lemma}
\label{lemma:implement}
    Let $H$ be an involution-free graph with $V(H) =
\{h_1, \dots, h_m\}$. 
    Every vector in
    $\GFtwo^{m}$ that is consistent for~$H$ is implementable for~$H$.
\end{lemma}

Suppose that~$G$ and~$H$ are graphs. 
We say that a pinning function $\pin\colon V(G) \to 2^{V(H)}$
is \emph{$r$-restrictive}
if there is a set $\{x_1,\ldots,x_r\}\subseteq V(G)$ 
such that 
\begin{itemize}
\item  for every $i\in[r]$ there is 
subset $W_i \subseteq V(H)$
such that $\pin(x_i) = \bigcup_{h\in W_i} \Orb_H(h)$, and
\item for every $w\in V(G)\setminus \{x_1,\ldots,x_r\}$, $\pin(w)=V(H)$.
\end{itemize}
Consider the following computational problem, which is parameterised
by a graph~$H$ and a natural number~$r$.
\begin{description}
\item  \emph{Name:} $\rpinnedparhcol$.
\item \emph{Input:} A graph $G$ and a $r$-restrictive pinning function $\pin\colon V(G) \to 2^{V(H)}$. 
\item \emph{Output:} ${\left|\HomPin(G,H,\pin)\right|} \imod 2$.
\end{description} 

The following theorem suffices for our purposes.
The same proof would establish a slightly more general result.
Instead of fixing the parameter~$r$, we could allow $r$ to depend on
$|V(G)|$, as long as $r = {\mathrm O}(\log |V(G)|)$.
However, we do not need this generalisation.

\begin{theorem}
\label{thm:pinning}
Let $H$ be an involution-free graph and let $r$~be a positive integer.
There is a polynomial-time Turing reduction from \rpinnedparhcol{} to 
$\parhcol{H}$.
\end{theorem}
\begin{proof}

Consider an instance $(G,\pin)$ of \rpinnedparhcol{} with $|V(H)|=m$ and
let $\{x_1,\ldots,x_r\}$ be the set of vertices restricted by the $r$-restrictive pinning function~$\pin$.
     For $i\in [r]$, let $\bv_i\in
    \{0,1\}^m$ be the characteristic vector of $\pin(x_i)$ (i.e., the vector
    that has 1s in positions corresponding to elements of $\pin(x_i)$ and 0s
    in every other position).  By Lemma~\ref{lemma:implement}, for
    each~$i$, there are gadgets $(\Theta_{i,1},z_{i,1}), \dots,
    (\Theta_{i,k_i}, z_{i,k_i})$ that implement $\bv_i$ for~$H$.  
 We
 assume that~$G$, $H$ and all the graphs $\Theta_{i,j}$ are
 pairwise disjoint. 

    Let $J = [k_1]\times \dots \times [k_r]$.  For any $\bj=(j_1,
    \dots, j_r)\in J$, let $G(\bj)$ be the graph formed from $G\cup
    \Theta_{1,j_1} \cup \cdots \cup \Theta_{r,j_r}$ by identifying
    $x_i$ with $z_{i,j_i}$ for each $i\in [r]$.

    We have
    \begin{align*}
        \sum_{\bj\in J} |\Hom(G(\bj),H)|
            &= \sum_{\bj\in J}
                  \sum_{\phi\in\Hom(G,H)}
                      \prod_{i\in [r]} \left|\left\{
                          \psi\in\Hom(\Theta_{i,j_i},H)\mid
                              \psi(z_{i,j_i}) = \phi(x_i)
                      \right\}\right|                         \\
            &= \sum_{\phi\in\Hom(G,H)}
                   \sum_{\bj\in J}
                       \prod_{i\in [r]} \left|\left\{
                           \psi\in\Hom(\Theta_{i,j_i},H)\mid
                               \psi(z_{i,j_i}) = \phi(x_i)
                       \right\}\right|                         \\
            &= \sum_{\phi\in\Hom(G,H)}
                   \prod_{i\in [r]}
                       \sum_{j\in [k_i]} \left|\left\{
                           \psi\in\Hom(\Theta_{i,j},H)\mid
                               \psi(z_{i,j}) = \phi(x_i)
                       \right\}\right|\,.                         \\
\intertext{Let $\indicator_{a\in A}=1$ if $a\in A$ and equal~0,
otherwise.
Since the gadgets $(\Theta_{i,j},z_{i,j})$ implement the
characteristic vector $v_i$ of~$\pin(x_i)$, we have}
        \sum_{\bj\in J} |\Hom(G(\bj),H)|
            &\equiv \sum_{\phi\in\Hom(G,H)}
                        \prod_{i\in [r]}
                            \indicator_{\phi(x_i)\in \pin(x_i)}\imod{2} \\
            &\equiv \left|\left\{
                        \phi\in\Hom(G,H)
                            \mid \phi(x_i)\in \pin(x_i) \text{ for all $i\in[r]$}
                    \right\}\right|\imod{2}\\
& \equiv {\left|\HomPin(G,H,\pin)\right|} \imod{2}\,,                 
                        \end{align*}
So the output corresponding to $(G,p)$ is   $   \sum_{\bj\in J} |\Hom(G(\bj),H)|$,
and the reduction consists of computing $J$, and for each $\bj \in J$
constructing $G(\bj)$ and summing the oracle's answers.

    We now consider the complexity of  the reduction.  Let $b$~be
    the greatest number of gadgets required to implement a pin to any
union of orbits of vertices of~$H$     
    and let $c$ be the greatest size of all the
    gadgets used to implement any such pinning.  
    Since $H$ is a fixed
    parameter, $b$ and~$c$ are constants, independent of the input
    graph~$G$. 
Also, $\max_{i\in[r]}k_i \leq b$ and $\max_{i\in[r],j\in[k_i]} |V(\Theta_{i,j})|\leq c$.

    Computing $\sum_{\bj\in J} |\Hom(G(\bj),H)|$ requires $|J|\leq b^r $ oracle calls.
     Each of these oracle calls applies to a graph $G(\bj)$ with at most $n + r c$ 
    vertices, each of which can be trivially constructed in
    polynomial time.  Therefore, we have a polynomial-time Turing
    reduction, as claimed.
\end{proof}


\section{Gadgets} 

In this section, we define the gadgets we use to prove
\parp{}-hardness of $\parhcol{H}$ problems by reduction from \paris.
These gadgets should not be confused with the pinning gadgets of the
previous section though, in the rest of the paper, we only use the
pinning gadgets implicitly, through Theorem~\ref{thm:pinning}.

\begin{definition} 
\label{def:hgad}
A \emph{hardness gadget} in a graph $H$
is a tuple $(\beta,s,t,O,i,K,k,w)$ 
where $\beta$ is a positive integer,
$s$, $t$ and~$i$ are vertices of $H$,
$(O,\{i\},K)$ is a partition of $\Gamma_H(s)$, and
$k\colon K \to \posints$ and $w\colon K \to V(H)$
are functions.
The following conditions must be satisfied.
\begin{enumerate}
\item\label{hgad1} $|O|$ is odd.
\item\label{hgad2} For any $o\in O$ and $y\in O\cup \{i\}$, $s$~is the
  unique vertex that is adjacent to $o$ and~$y$ and has an odd number
  of $\beta$-walks to~$t$.
\item\label{hgad3} There are an even number of $(1+\beta)$-walks from
  $i$ to~$t$.
\item\label{hgad4} For all $u\in K$, $w(u)$ has an even number of $k(u)$-walks to~$u$ 
and an odd number of $k(u)$-walks to every vertex in $O\cup \{i\}$.
\end{enumerate} 
\end{definition}

These conditions simplify if $\beta=1$, since having an odd number of
$1$-walks to a vertex is the same as being adjacent to it.  In cases where
$\beta=1$, we will use this simplified condition without comment.

The construction used in our reduction from \paris{} is given formally
in Definition~\ref{def:G-Gamma}.  Given a graph~$G$ and a hardness
gadget~$\Gamma\!$, we will produce a graph $G_\Gamma$ that includes a
copy of~$V(G)$.  We call the vertices in this copy, ``$G$-vertices''.
We will consider homomorphisms from $G_\Gamma$ to $H$.  Using the
pinnings described in the previous section, we will restrict attention
to homomorphisms that map all $G$-vertices to neighbours of~$s$.
Part~\ref{hgad4} of the definition ensures that there will be an even
number of such homomorphisms that map any $G$-vertices to members
of~$K$.  These contribute nothing to the total modulo~2 so the effect
is to restrict to homomorphisms that map every $G$-vertex to
$O\cup\{i\}$.  Part~\ref{hgad3} of the definition will ensure that the
number of homomorphisms that map adjacent vertices in~$G$ to~$i$ is
even so these also do not contribute.  Thus, the homomorphisms that
remain are those in which an independent set of $G$-vertices mapped
to~$i$.  Our key technical result is that every
non-trivial, involution-free cactus graph contains a hardness gadget
(Theorem~\ref{thm:hardness-gadget}).

Our analysis is based primarily on decompositions of graphs into
their subgraphs, so we need conditions under which a hardness gadget in an
induced subgraph of~$H$ is a hardness gadget in~$H$.  The idea here is
that, if a hardness gadget satisfies the distance requirements for a
vertex~$v$, the structure of the graph ``beyond''~$v$ cannot
interfere with the gadget's paths.
 
\begin{definition}  
Consider a hardness gadget $(\beta,s,t,O,i,K,k,w)$ in~$H$ and a vertex $v\in V(H)$.
The \emph{primary distance requirement} of the gadget with respect to~$v$
is
$$\dist_H(v,O\cup \{i\}) + \dist_H(v,t) > \beta-1.$$
The \emph{secondary distance requirement} of the gadget with respect to~$v$
is that, for each $u\in K$,
$$\dist_H(v,w(u))+\dist_H(v,O\cup \{i,u\}) >k(u)-2.$$
\end{definition}

Suppose that $H_1, \dots, H_\kappa$ is a split of the graph~$H$ at some cut
vertex~$v$.  If there is a hardness gadget~$\Gamma$ in $H_1$ that
satisfies the distance restrictions for~$v$, it is easy to see that it
also satisfies the distance restrictions for all $x\in V'=V(H_2) \cup
\dots \cup V(H_\kappa)$, since any path from $H_1$ to~$V'$ must go
through~$v$.  This ensures that $\Gamma$~is also a hardness gadget
in~$H$, since the number of walks of various lengths required by the
definition of the hardness gadget cannot be affected by vertices
beyond~$v$.

In some cases, our decomposition might yield subgraphs that do not
contain hardness gadgets.  We are still able to make progress using
structures that can be combined with other parts of the graph to
produce a hardness gadget.  A partial hardness gadget is, essentially,
a simplified hardness gadget that has $K=\emptyset$ and that doesn't
yet have a ``$t$'' vertex: at a later point, we will find a vertex~$t$
with the properties necessary to produce a full hardness gadget.

\begin{definition} A \emph{partial hardness gadget} in a rooted graph $(H,x)$  
is a tuple $(s,i,O,P)$
where  
$s$ is a vertex of~$H$,
$(\{i\},O)$ is a partition of $\Gamma_H(s)$, and $P$ is a path in~$H$ 
satisfying the following conditions.
\begin{itemize}
\item $|O|$ is odd.
\item $P$ is the unique shortest path from $x$ to $i$ in $H$.
\item $Ps$ is the unique shortest path from $x$ to $s$ in $H$.
\item For each $o\in O$, $Pso$
is the unique shortest path from $x$ to $o$ in $H$.
\end{itemize}
\end{definition}


\section{Mosaics}

The final structure that can arise from our decompositions is a
subgraph made entirely from 4-cycles and from edges between vertices of
those cycles and additional vertices of degree~1.  Some mosaics (the ``shortcut
mosaics'' defined below) already contain hardness gadgets.  In the
other cases, we identify structures called ``2,3-paths'' in mosaics
and these will provide a ``$t$''~vertex for a partial hardness gadget
elsewhere in the decomposed graph.

\begin{definition} 
 A \emph{mosaic} is a rooted graph  defined as follows.  
\begin{itemize}
\item An unbristled mosaic is the one-vertex rooted graph or a rooted cactus graph that is a union of 4-cycles.
\item A \emph{mosaic}
is a rooted graph $(H,x)$ 
for which there is a 
partition $(V'\!,V'')$ of $V(H)$ such that $x$~is in~$V'\!$,
$(H[V'],x)$ is an unbristled mosaic,
and  $E(H)-E(H[V'])$ is a perfect matching between
$V''$ and a subset of~$V'\!$.
The edges of the matching are called \emph{bristles}.
\item A \emph{proper mosaic} is a
mosaic that contains at least one cycle.
\end{itemize}
\end{definition}

The graphs in Figure~\ref{fig:hard-easy-hard} would be mosaics if a
root were placed at any vertex on a cycle.
Note that every vertex of a mosaic is adjacent to at most one bristle.
Note also that the one-vertex rooted graph and a rooted edge are both
mosaics (but not proper mosaics).

\begin{definition}
    A \emph{2,3-path} in a rooted graph $(H,x)$ is a tuple
    $(P,v_2,v_3)$ such that $v_2$ and $v_3$ are in the same cycle
    of~$H$ and that, for $j\in\{2,3\}$, $\deg_H(v_j)=j$ and
    $Pv_j$ is the unique shortest $x$--$v_j$ path in~$H$.
\end{definition}

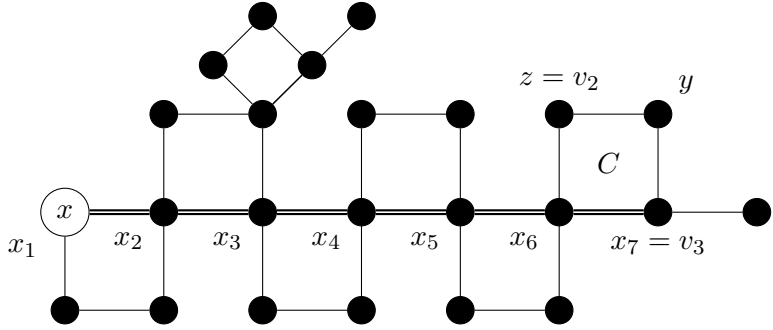
\begin{figure}[t]
\begin{center}
\begin{tikzpicture}[scale=0.65]
        \tikzstyle{vertex}=[fill=black, draw=black, circle]
	\tikzstyle{root}=[ draw=black, minimum size=1mm ,circle]

\node[root] (B1) at (0,0)[label=225:$x_1$] {$x$};

\foreach \x in {2,3,4,5} \node[vertex] (A\x) at (2*\x-2,2) {};
\node[vertex] (A6) at (10,2) [label=90:{$z=v_2$}] {};
\node[vertex] (A7) at (12,2) [label=45:$y$] {};

\foreach \x in {2,3,4,5,6} \node[vertex,label=225:$x_\x$] (B\x) at (2*\x-2,0) {};
\node[vertex] (B7) at (12,0) [label=-90:{$x_7=v_3$}] {};  
\node[vertex] (B8) at (14,0) {};
\node (C) at (11,1) {$C$};

\node[vertex] (e1) at (5,3) {};
\node[vertex] (e2) at (3,3) {};
\node[vertex] (e3) at (4,4) {};
\node[vertex] (e4) at (6,4) {};

\draw (e4)--(e1)--(A3)--(e2)--(e3)--(e1)--(A3);

\foreach \x in {1,2,3,4,5,6} \node[vertex] (C\x) at (2*\x-2,-2) {};

\foreach \x in {2,3,4,5,6} \draw (A\x)--(C\x);

\draw[thick, double] (B1)--(B2)--(B3)--(B4)--(B5)--(B6)--(B7);

\draw (B7)--(B8);
\draw (B1)--(C1);
\draw (A7)--(B7);

\draw (A2)--(A3);
\draw (A4)--(A5);
\draw (A6)--(A7);

\draw (C1)--(C2);
\draw (C3)--(C4);
\draw (C5)--(C6);

\end{tikzpicture}
\caption{A mosaic with the 2,3-path $(P, v_2, v_3)$ that can be found using Lemma~\ref{lemma:mosaic-deg23}. $P$ is drawn with a double line.\label{fig:deg23}}
\end{center}
\end{figure}

\begin{lemma}\label{lemma:mosaic-deg23} 
Every involution-free proper mosaic $(H,x)$ contains a 2,3-path.
 \end{lemma}
\begin{proof}
See Figure~\ref{fig:deg23}.
Write $x_1=x$ and let $P = x_1\dots x_\ell$ be a longest path 
from~$x_1$
in $H$ that uses only edges from cycles and uses at most one
edge from each cycle. $P$ contains at least one edge because $x$ is on a cycle. Let $C=x_{\ell-1}x_\ell yzx_{\ell-1}$ be the
cycle containing $x_{\ell-1}$ and~$x_\ell$.

$x_\ell$ is on a cycle, so $\deg(x_\ell)\geq 2$. Also, $\deg(x_\ell)\leq 3$ since, otherwise, $x_\ell$ would have a neighbour
$x_{\ell+1}$ on a cycle other than~$C$ and the path $Px_{\ell+1}$ would contradict the choice of $P$. By the same argument,
$2\leq\deg(z)\leq 3$. Furthermore, $\deg (x_\ell)\neq \deg (z)$ or 
$H$ would have an
involution  exchanging these two
vertices. Note that $P'=x_1\dots x_{\ell-1}z$ is the unique shortest $x$--$z$ path in $H$, since any other $x$--$z$ path
must include edges from exactly the same cycles as $P'$ and must
include at least two edges from one of them.
Thus, either $(x_1\dots x_{\ell-1},x_\ell,z)$ or $(x_1,\dots,x_{\ell-1},z,x_\ell)$
is a 2,3-path.
\end{proof}

In several cases, we have a unique shortest path of length~$\ell$
between two vertices in a graph and we are interested in the number of
$(\ell+2)$ walks between those two vertices.  The following definition
helps us count such walks.

\begin{figure}[t]
\begin{center}
\begin{tikzpicture}[scale=0.85]  
        \tikzstyle{vertex}=[fill=black, draw=black, circle]
	\tikzstyle{root}=[ draw=black, minimum size=1mm ,circle]

\node[vertex] (1) at (2,0) [label=-90:$x_a$] {};
\node[vertex] (2) at (4,0) {};
\node[vertex] (3) at (6,0) [label=-90:$x_b$]{};
\foreach \x in {1,2,3} \node[vertex] (c\x) at (2*\x,2) {};

\node[vertex] (0) at (0,0) [label=-90:$x_1$] {};
\node[vertex] (4) at (8,0) [label=-90:$x_{\ell+1}$] {};

\draw (1)--(2)--(3);

\draw[decorate, decoration=snake, thick, double] (0)--(1);
\draw[decorate, decoration=snake, thick, double] (3)--(4);

\draw[thick, double] (1)--(c1)--(c2)--(c3)--(3);

\draw[decorate, decoration=brace, thick] (6,-1)--(2,-1);
\node at (4,-1.5) {\small{$r$}};
\end{tikzpicture}
\caption{An example of a walk in $T_1(P)$, shown with double lines.\label{fig:plustwo-cycle}}
\end{center}
\end{figure}
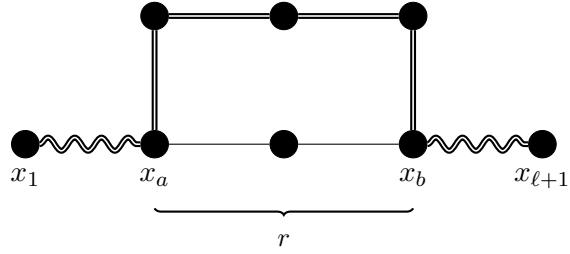

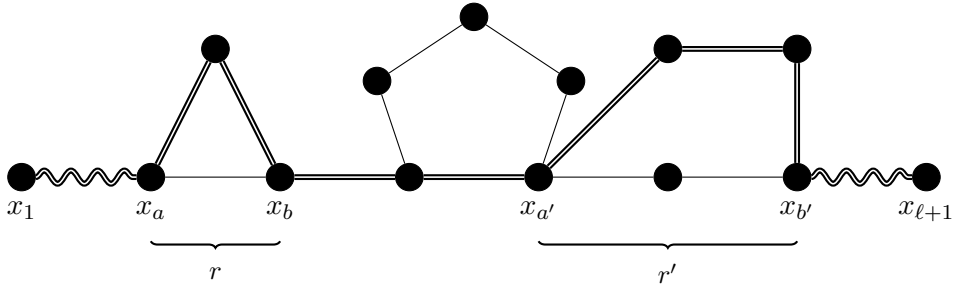
\begin{figure}[t]
\begin{center}
\begin{tikzpicture}[scale=0.85]  
        \tikzstyle{vertex}=[fill=black, draw=black, circle]
	\tikzstyle{root}=[ draw=black, minimum size=1mm ,circle]

\node[vertex] (0) at (0,0) [label=-90:$x_1$] {};
\node[vertex] (1) at (2,0) [label=-90:$x_a$] {};
\node[vertex] (2) at (4,0) [label=-90:$x_b$] {};
\node[vertex] (3) at (6,0) {};
\node[vertex] (4) at (8,0) [label=-90:$x_{a'}$] {};
\node[vertex] (5) at (10,0) {};
\node[vertex] (6) at (12,0) [label=-90:$x_{b'}$] {};
\node[vertex] (7) at (14,0) [label=-90:$x_{\ell+1}$] {};
\node[vertex] (c1) at (3,2) {};
\node[vertex] (c2) at (10,2) {};
\node[vertex] (c3) at (12,2) {};

\node[vertex] (a) at (8.5,1.5) {};
\node[vertex] (b) at (5.5,1.5) {};
\node[vertex] (c) at (7,2.5) {};

\draw (1)--(2);
\draw (4)--(5)--(6);
\draw (3)--(b)--(c)--(a)--(4);

\draw[decorate, decoration=snake, thick, double] (0)--(1);
\draw[decorate, decoration=snake, thick, double] (6)--(7);

\draw[thick, double] (1)--(c1)--(2)--(3)--(4)--(c2)--(c3)--(6);

\draw[decorate, decoration=brace, thick] (4,-1)--(2,-01);
\node at (3,-1.5) {\small{$r$}};
\draw[decorate, decoration=brace, thick] (12,-1)--(8,-1);
\node at (10,-1.5) {\small{$r'$}};
\end{tikzpicture}
\caption{An example of a walk in $T_2(P)$, shown with double lines.\label{fig:plusone-cycle}}
\end{center}
\end{figure}

\begin{definition}
\label{def:T-walks}
Let $P=x_1\dots x_{\ell+1}$ be a path in a graph~$H$.  We define the
following three sets of $(\ell+2)$-walks from $x_1$ to $x_{\ell+1}$.
\begin{enumerate}
\item $T_1(P)$ is the set of walks that differ from~$P$ by going the
  long way around a $(2r+2)$-cycle from which $P$ uses $r$~consecutive
  edges, as shown in Figure~\ref{fig:plustwo-cycle}.  Formally, these
  are the walks $x_1\dots x_aP'x_b\dots x_{\ell+1}$, where $1\leq
  a<b\leq \ell+1$ and $x_aP'x_b$ is a $(b-a+2)$-path in $H-\{x_{a+1},
  \dots, x_{b-1}\}$.
\item $T_2(P)$ is the set of walks that differ from~$P$ by going the
  long way around cycles of length $2r+1$ and $2r'+1$, from which
  $P$~uses $r$ and~$r'$ consecutive edges, respectively, as shown in
  Figure~\ref{fig:plusone-cycle}.  Formally, these are the $x_1\dots
  x_aP'x_b\dots x_{a'}P''x_{b'} \dots x_{\ell+1}$ where $1\leq a<b\leq
  a'<b'\leq \ell+1$, $x_aP'x_b$ is a $(b-a+1)$-path in $H-\{x_{a+1},
  \dots, x_{b-1}\}$ and $x_{a'}P''x_{b'}$ is a $(b'-a'+1)$-path in
  $H-\{x_{a'+1}, \dots, x_{b'-1}\}$.
\item $T_3(P)$ is the set of walks $x_1\dots x_azx_a\dots x_{\ell+1}$,
  where $1\leq a\leq \{\ell+1\}$ and $z\in\Gamma_H(x_a)$ (we allow the
  case $z=x_{a\pm 1}$).
\end{enumerate}
We refer to the cycles appearing in the definition of $T_1$ and~$T_2$
as \emph{detour cycles}.
\end{definition}

It is easy to see that, when $P$ is the unique shortest
$x_1$--$x_{\ell+1}$ path in a cactus graph, $T_1(P)$, $T_2(P)$ and
$T_3(P)$ is a partition of the set of all $(\ell(P)+2)$-walks from $x_1$
to~$x_{\ell+1}$.

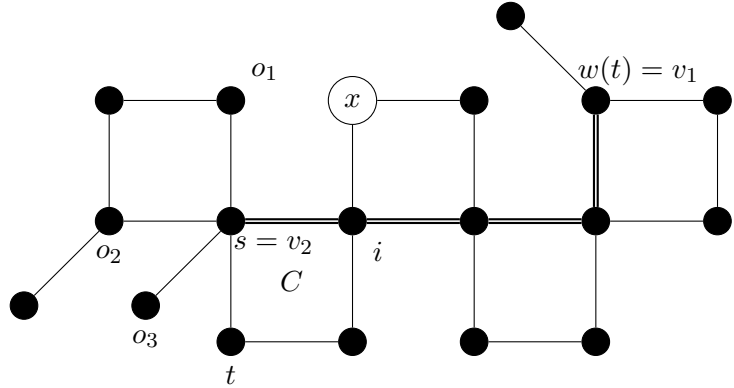
\begin{figure}[t]
\begin{center}

\begin{tikzpicture}[scale=.8]
        \tikzstyle{vertex}=[fill=black, draw=black, circle]
	\tikzstyle{root}=[ draw=black, minimum size=1mm ,circle]

\foreach \x in {6,7,8}  \node[vertex] (\x) at (2*\x-1,0) {};
\node[vertex] (3) at (5,0) [label=-90:$o_2$] {};
\node[vertex] (4) at (7,0) {};\node at (7.7,-.4) {$s=v_2$};
\node[vertex] (5) at (9,0) [label=-45:$i$] {};
\foreach \x in {3,6,8} \node[vertex] (a\x) at (2*\x-1,2) {};
\node[vertex] (a4) at (7,2) [label=45:$o_1$] {};
\node[root] (a5) at (9,2) {$x$};
\node[vertex] (a7) at (13,2) {}; \node at (13.7,2.5) {$w(t)=v_1$};
\foreach \x in {5,6,7}  \node[vertex] (b\x) at (2*\x-1,-2) {};
\node[vertex] (b4) at (7,-2) [label=-90:$t$] {};

\node[vertex] (a) at (3.6,-1.4) {};
\node[vertex] (b) at (11.6,3.4) {}; \draw (a7)--(b);
\node[vertex] (c) at (5.6,-1.4) [label=below:$o_3$] {}; \draw (4)--(c);

\draw (a)--(3);
\draw (3)--(a3)--(a4)--(4)--(b4)--(b5)--(5)--(a5)--(a6)--(6)--(b6)--(b7)--(7);
\draw (a7)--(a8)--(8)--(7);
\draw (3)--(4);
\draw[double, thick] (4)--(5)--(6)--(7)--(a7);
\node (C) at  (8,-1)  {$C$};
\end{tikzpicture}
\caption{A mosaic containing containing a hardness gadget as in Lemma~\ref{lemma:shortcut-hardness}.
The path $P$ is shown with a double line and $O=\{o_1,o_2,o_3\}$.\label{fig:mosaic-shortcut}}
\end{center}
\end{figure}

\begin{lemma}\label{lemma:shortcut-hardness} 
Let $H$ be a cactus graph containing distinct odd-degree
vertices $v_1$ and $v_2$, with a unique shortest path~$P$ between them.
Suppose that every edge of~$P$ is on a $4$-cycle of~$H$ and no
$4$-cycle of~$H$ contains two edges of~$P$.
Then $H$ contains a hardness gadget that satisfies the distance requirements
for every $v \in (V(H)\setminus V(P)) \cup \{v_2\}$. 
\end{lemma}

Note that the premise of the lemma is symmetric about $v_1$ and~$v_2$,
while the conclusion is not.  By symmetry, we could, of course, find a
different hardness gadget satisfying the distance requirements for
$v_1$ instead of~$v_2$.

\begin{proof}
See Figure~\ref{fig:mosaic-shortcut}.
Choose $v_1$, $v_2$ and $P$ satisfying the given conditions
so that $\ell(P)$ is as small as possible.
Note that $v_1$ and~$v_2$ have degree at least~$3$, since they have
odd degree and are on cycles.
  
Let $\beta=1$.
Let $s=v_2$, 
let $i$ be the neighbour of $s$ in $P$ and let $t$ be the neighbour
of $s$ that is in the same $4$-cycle as $i$; call this cycle~$C$.
Let $K=\{t\}$, and $O=\Gamma_H(s)\setminus (\{i\} \cup K)$.
Let $w(t)=v_1$ and $k(t) = \ell(P)+1$. 

To see that $\Gamma=(\beta,s,t,O,i,K,k,w)$ is a hardness gadget,
note first that $|O|$ is odd, since $\deg_H(s)$ is odd.
Consider any $o\in O$ and $y\in O \cup \{i\}$. Since $H$ is a cactus graph,
$s$~is the unique vertex adjacent to $o$, $y$ and~$t$.
However, there are two vertices
that are adjacent to~$i$ and~$t$:
$s$ and the 
fourth vertex on cycle~$C$. Therefore, there are two $(\beta+1)$-walks
from $i$ to~$t$.
Finally, consider vertex $w(t)=v_1$.
Vertex~$v_1$ has two $(\ell(P)+1)$-walks to~$t$ (going around the cycle~$C$ in either direction).
For every vertex $o\in O$, 
$Po$ is the unique $(\ell(P)+1)$-walk from $v_1$ to~$o$.

We will finish the proof that $\Gamma$ is a hardness gadget by showing
that $w(t)=v_1$ has an odd number of $(\ell(P)+1)$-walks to $i$.
Let $P_i$ be the length-($\ell(P)-1$) prefix of~$P$.
$P_i$ is the unique shortest path from~$v_1$ to~$i$.
Consider Definition~\ref{def:T-walks}, with $x_1=v_1$,
$\ell=\ell(P_i)$ and $x_{\ell+1}=i$.  Since $P_i$ only uses one edge
from each cycle it meets and each such cycle is a $4$-cycle,
$T_2(P_i)=\emptyset$.  
There are $\ell(P_i)$ walks in $T_1(P_i)$,
since every edge of~$P$ (hence, every edge of~$P_i$) is on a distinct $4$-cycle of~$H$.
The number of walks in $T_3(P_i)$ is
$\sum_{v \in P_i} \deg_H(v) - \ell(P_i)$
since every edge adjacent to~$P_i$ may be repeated, but
edges in $P_i$ should not be counted twice.
 The total number of walks
is therefore $\sum_{v \in P_i} \deg_H(v)$. This is odd since $\deg_H(v_1)$
is odd, and every vertex in~$P_i$ other than~$v_1$ has even degree
(otherwise the minimality of $\ell(P)$ would be contradicted).

The hardness gadget $\Gamma$ satisfies the primary distance requirement
$\dist_H(v,O\cup \{i\}) + \dist_H(v,t) > 0$ for any $v\in V(H)$
since $t\not\in O\cup \{i\}$, so at least one of the terms $\dist_H(v,O\cup \{i\})$ and $\dist_H(v,t)$ is positive.
Now consider any $v \in V(H)\setminus V(P_i)$.
We wish to show that the secondary distance requirement
$\dist_H(v,v_1)+\dist_H(v,O\cup \{i,t\})>\ell(P)-1$ is
satisfied.
We do this by establishing the following inequalities:
\begin{gather}
\label{eqone} \dist_H(v,v_1)+\dist_H(v,i)>\ell(P)-1\\
\label{eqtwo} \dist_H(v,v_1)+\dist_H(v,O\cup \{t\})>\ell(P)-1.
\end{gather}
Establishing \eqref{eqone} is easy.
Since $v\not\in P_i$ and $P_i$ is the unique shortest path from~$v_1$ to~$i$,
$\dist_H(v,v_1) + \dist_H(v,i) > \ell(P_i)=\ell(P)-1$.
Establishing \eqref{eqtwo} is similar.
For each $y\in O\cup \{t\}$,  
$\dist_H(v_1,y) = \ell(P)+1$.
Therefore, 
for any $v\in V(H)\setminus V(P_i)$,
$\dist_H(v,v_1)+\dist_H(v,y)\geq \ell(P)+1$.
\end{proof}

\begin{definition}  
A \emph{shortcut} in a mosaic $(H,x)$  
is a pair of odd-degree vertices $v_1,v_2$, with degree at least~3, that
have a unique shortest path $P$ between them, and this path does not
contain~$x$.
A \emph{shortcut mosaic} is a mosaic that contains a
shortcut.
\end{definition}

If $v_1,v_2$ is a shortcut in a mosaic $(H,x)$ then $v_1$ and $v_2$ are on cycles
(since their degrees are at least~$3$), so every edge of the unique shortest path~$P$ between them
is on a $4$-cycle. Since $P$ is unique, these edges are on distinct $4$-cycles. Thus,
Lemma~\ref{lemma:shortcut-hardness} has the following corollary.

\begin{corollary}\label{corollary:mosaic-hardness}
If $(H,x)$ is a shortcut mosaic then $H$ contains a hardness gadget that satisfies the
distance requirements for $x$.
\end{corollary}


\section{Combination lemmas}

We mostly proceed by splitting graphs at cut vertices and
investigating the resulting components.  In this section, we present a
number of technical lemmas that show how
to combine structures in the various parts of a graph split to obtain
hardness gadgets.
 
\begin{observation}
\label{obs:rootsplit}
  If $\{H_1, \dots, H_\kappa\}$ is the split of an
   involution-free graph~$H$ at a cut vertex~$v$ then, for each $j\in[\kappa]$, the rooted
graph $(H_j,v)$ is involution-free even though $H_j$ itself might not be involution-free.
To see that $(H_j,v)$ is involution-free, note that an involution of $H_j$ that fixes $v$ 
induces an involution of $H$.
 \end{observation}

\begin{lemma}\label{lemma:mosaic-mosaic-hardness}
Let $x$ be a cut vertex of an involution-free cactus graph $H$. If there exists a split of $H$ at $x$ into
$\{H_1,\dots,H_\kappa\}$ 
such that $(H_1,x)$ and $(H_2,x)$ are both proper mosaics then $H$ has a hardness gadget
which satisfies the distance requirements   
for every vertex $v\in V(H) \setminus (V(H_1)\cup V(H_2))$.
\end{lemma}
\begin{proof}
If, for $j=1$ or $j=2$, $(H_j,x)$ is a shortcut mosaic then, by  Corollary~\ref{corollary:mosaic-hardness},
$H_j$ contains a hardness gadget~$\Gamma$ that satisfies the
distance requirements for~$x$.  Since $x$~is a cut vertex, $\Gamma$~is
a hardness gadget in~$H$ and satisfies the distance restrictions for
every vertex outside $H_1$ and $H_2$.

Suppose that neither of $(H_1,x)$ and $(H_2,x)$ is a shortcut mosaic.
For $j\in\{1,2\}$, apply Lemma~\ref{lemma:mosaic-deg23} to $(H_j,x)$
to obtain a 2,3-path $(P_j, y_j, z_j)$.  Since $z_j\neq x$ and $x$~is
a cut vertex, $\deg_H(z_j)=\deg_{H_j}(z_j)=3$.  $z_1P_1P_2z_2$ is the unique
shortest $z_1$--$z_2$ path in~$H$ and each edge of this path is in a
different 4-cycle.  
By Lemma~\ref{lemma:shortcut-hardness}, $H$ contains a hardness gadget
which satisfies the distance requirements for every $v \in ( V(H) \setminus V(P)) \cup \{z_2\}$.
\end{proof}

\begin{lemma} \label{lemma:mosaic-oddroot}    
Let $H$ be a cactus graph with a cut vertex $x$ of odd degree.
 If there exists a split of
$H$ at $x$ into $\{H_1,\dots,H_\kappa\}$ such that $(H_1,x)$ is an involution-free proper mosaic, 
then $H$ has a hardness gadget
that satisfies the distance requirements for  every $v \in V(H)\setminus V(H_1)$.
\end{lemma}

\begin{proof} 
If $(H_1,x)$ is a shortcut mosaic then Corollary~\ref{corollary:mosaic-hardness}
gives a hardness gadget in $H_1$. As in the previous lemma, this is a
hardness gadget in~$H$ and satisfies the distance requirements.
If $(H_1,x)$ is not a shortcut mosaic,
apply Lemma~\ref{lemma:mosaic-deg23} to
$(H_1,x)$ to obtain a 2,3-path $(P,v_2,v_3)$.
Since $v_3\neq x$ and $x$~is a cut vertex,
$\deg_H(v_3)=\deg_{H_1}(v_3)=3$ and $P$ is the unique shortest path in~$H$ between~$v_3$ and $x$. 
Since~$(H_1,x)$ is a proper mosaic, every edge of~$P$ is on a distinct $4$-cycle of~$H$.
By Lemma~\ref{lemma:shortcut-hardness}, $H$ contains a hardness gadget
that satisfies the distance requirements for every $v \in ( V(H) \setminus V(P)) \cup \{x\}$, which proves the lemma.
\end{proof}

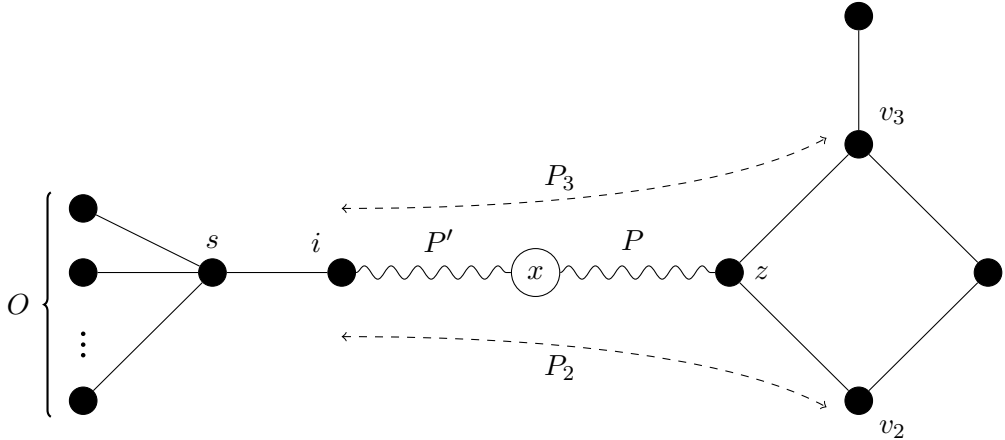
\begin{figure}[t]
\begin{center}
\begin{tikzpicture}[scale=.85]
        \tikzstyle{vertex}=[fill=black, draw=black, circle]
	\tikzstyle{root}=[ draw=black, minimum size=1mm ,circle]
	\node[vertex] at (-2, 0) (0) {}; 
        \node[vertex] at (0, 0) (3) [label=90:$s$] {};
	\node[vertex] at (2, 0) (4)[label=135:$i$]{};
        \node[root]   at (5, 0) (5) {$x$};
        \node[vertex] at (8, 0) (6) [label=0:$z$] {};
        \node[vertex] at (10,-2) (7) [label=-45:$v_2$] {};
        \node[vertex] at (12, 0) (8) {};
        \node[vertex] at (10, 2) (9) [label=45:$v_3$] {};
	\node[vertex] at (10, 4) (10) {};
	\node[vertex] at (-2, 1) (1) {}; 
	\node[vertex] at (-2, -2) (2) {}; 

\draw[decorate, decoration=brace, thick] (-2.5,-2.25)--(-2.5,1.25);
\node at (-3,-0.5) {$O$};

\draw (9)--(10);
\draw (3)--(4);
\draw (6)--(7)--(8)--(9)--(6);

\foreach \y in {0,1,2} \path[-] (\y) edge (3);
\draw (-2,-1) node {\huge{$\vdots$}};

\draw[decorate, decoration=snake] (4)--(5);
\draw[decorate, decoration=snake] (5)--(6);

\draw[<->, dashed] (2,-1) .. controls (3,-1) and (7.5,-1) .. node[below] {$P_2$} (9.5,-2.1);

\draw[<->, dashed] (2,1) .. controls (3,1) and (7.5,1) .. node[above] {$P_3$} (9.5,2.1);
\node (P) at (3.5,0.5) {$P'$};
\node (Pd) at (6.5,0.5) {$P$};
\end{tikzpicture}
\caption{The paths appearing in the proof of Lemma~\ref{lem:phg-23path}.\label{fig:mosaic-phg-hardness}}
\end{center}
\end{figure} 

\begin{lemma} \label{lem:phg-23path}
Let $H$ be a cactus graph. 
Suppose that $\{H_1,\ldots,H_\kappa\}$ is the split of~$H$ at~$x$, that
$(H_1,x)$ contains a 2,3-path $(P, v_2, v_3)$
and that $(H_2,x)$ contains a partial hardness gadget $(s,i,O,P')$.
Then 
$H$ has
a hardness gadget
that satisfies the distance requirements for very vertex 
$v\in V(H) \setminus ( V(PP') \cup \{v_2,v_3\})$.
\end{lemma}
\begin{proof} 
See Figure~\ref{fig:mosaic-phg-hardness}.
Let $z$ be the endpoint of the path~$P$ that is not~$x$, or let $z=x$
if $\ell(P)=0$. 
Since $H$ is a cactus graph and $v_2$ and $v_3$ are on a cycle together,
the cycle must contain the edges $(z,v_2)$ and $(z,v_3)$.
Since $v_2\neq x$ and $v_3\neq x$ and $x$~is a cut vertex,
$\deg_H(v_2)=2$ and $\deg_H(v_3)=3$.  It is easy to see that $(P'P,
v_2, v_3)$ is a 2,3-path in~$H$.  For $j\in\{2,3\}$, let $P_j =
P'Pv_j$.

We next show that the number of $(\ell(P_2)+2)$-walks from~$i$ to~$v_2$ differs in parity
from the number of $(\ell(P_2)+2)$-walks from~$i$ to~$v_3$.
There are two kinds of walks to consider --- those that detour around cycles, and
those that repeat an edge.
Using the notation from Definition~\ref{def:T-walks}, walks that
detour around cycles are in $T_1$ or~$T_2$ and those that repeat an
edge are in~$T_3$.
Since $(z,v_2)$ and $(z,v_3)$ are not on any cycles other than the one that
contains them both, all of these walks pass through~$z$, progressing on to~$v_2$ or
to~$v_3$. Furthermore, the
number of $(\ell(P_2)+2)$-walks 
that detour around cycles is the same for both endpoints, $v_2$ and $v_3$.
Now note that the number of 
$(\ell(P_2)+2)$-walks in~$T_3(P_3)$
is exactly one more than the 
number in~$T_3(P_2)$, because $\deg_H(v_3)=\deg_H(v_2)+1$.
 
We can now define the hardness gadget.
Let $\beta=\ell(P_2)+1$. 
$s$, $O$ and~$i$ are already defined by the partial hardness gadget; let $K=\emptyset$.
Choose $t\in \{v_2,v_3\}$ so that the number of $(1+\beta)$-walks from~$i$ to~$t$ is even.

To see that this is a hardness gadget, consider $o\in O$ and $y\in O\cup \{i\}$.
$s$ is the only vertex 
that is both adjacent to~$o$ and~$y$ and has  an odd number of
$(\ell(P_2)+1)$-walks to~$t$ (otherwise, there would be more than
one shortest path from $o$ to~$t$).
By construction, there are an even number of $(1+\beta)$-walks from~$i$ to~$t$.  $K=\emptyset$ so
the requirements on~$K$ are vacuous.

Now consider the primary distance requirement
$\dist_H(v,O\cup \{i\}) + \dist_H(v,t) > \beta-1 = \ell(P_2)$. 
The unique shortest path in~$H$ from~$t$ to~$O\cup \{i\}$
is either $P_2$ or $P_3$ and this path has length
$\ell(P_2)$ 
so the requirement is satisfied for any vertex~$v$ that is not on
$P_2$ or $P_3$, which is to say, any vertex of $V(H)\setminus
(V(P)\cup V(P')\cup \{v_2,v_3\})$, as required.
There are no secondary distance requirements, since $K=\emptyset$, so the lemma is proved.
\end{proof}
 
\begin{corollary}\label{corollary:phg-mosaic-hardness} 
Let $H$ be an involution-free cactus graph and let $x$ be a cut vertex of $H$. 
If there exists a split of $H$ at $x$ into
$\{H_1,\dots,H_\kappa\}$ such that 
$(H_1,x)$ is a proper mosaic and $(H_2,x)$ contains a partial hardness gadget 
$(s,i,O,P)$,
then $H$ has
a hardness gadget
that satisfies the distance requirements for very vertex 
$v\in V(H) \setminus V(H_1\cup P)$.
\end{corollary}
\begin{proof}
By Lemma~\ref{lemma:mosaic-deg23}, $(H_1, x)$ contains a 2,3-path;
Lemma~\ref{lem:phg-23path} gives the hardness gadget.
\end{proof}

\begin{figure}[t]
\begin{center}
\begin{tikzpicture}[scale=0.85]  
        \tikzstyle{vertex}=[fill=black, draw=black, circle]
	\tikzstyle{root}=[ draw=black, minimum size=1mm ,circle]

	\node[vertex] at (-2, 0) (0) {};
        \node[vertex] at (0, 0) (5) [label=-90:{$s_1$}] {};
	\node[vertex] at (2, 0) (18)[label=90:{$i_1$}]{};
        \node[root]   at (5, 0) (6) {$x$};
        \node[vertex] at (12, 0) (0+1) {};
        \node[vertex] at (10, 0) (5+1) [label=-90:{$s_2$}] {};
	\node[vertex] at (8, 0) (18+1)[label=90:{$i_2$}]{};

\foreach \y in {1,2,3} \node[vertex] at (-2, 4-\y) (\y) {};
              
	\node[vertex] at (-2, -2) (4) {};
	\node[vertex] at (12, 1) (1+1) {};              
	\node[vertex] at (12, -2) (4+1) {};

\draw[decorate, decoration=brace, thick] (-2.5,-2.25)--(-2.5,3.25);
\node at (-3,0.5) {$O_1$};
\draw[decorate, decoration=brace, thick] (12.5,1.25)--(12.5,-2.25);
\node at (13,-0.5) {$O_2$};

 \path[-] (5) edge (18)
	  (5+1) edge (18+1);

\foreach \y in {0,1,2,3,4} \path[-] (\y) edge (5);

\foreach \y in {0,1,4} \path[-] (\y+1) edge (5+1);

\draw (-2,-1) node {\huge{$\vdots$}};
\draw (12,-1) node {\huge{$\vdots$}};
\draw [decorate, decoration=snake] (18) -- (6);
\draw [decorate, decoration=snake] (6)--(18+1);
\node at (3,0.5) {$P_1$};
\node at (7,0.5) {$P_2$};
	  \node[inner sep=0] at (2.5,4) (A) {}; 
	  \node[inner sep=0] at (7.5,4) (B) {}; 
\path[dashed, thick] (6) edge [bend left] (A)
              (B) edge [bend left] (6)
	      (A) edge [bend left] (B);
\draw (5,3) node  
{$v\notin V(H_1)\cup V(H_2)$};

\draw[decorate,decoration=brace,thick] (8,-.5)--(2,-.5);
\node at (5,-1) {\small{$\ell = \ell(P)$}};
\draw[decorate,decoration=brace,thick] (10,-1.5)--(2,-1.5);
\node at (6,-2) {\small{$\ell' = \ell(Ps_2)$}};
\end{tikzpicture}
\caption{The two partial hardness gadgets of Lemma~\ref{lemma:phg-phg-hardness}.\label{fig:phg-phg}}
\end{center}
\end{figure}
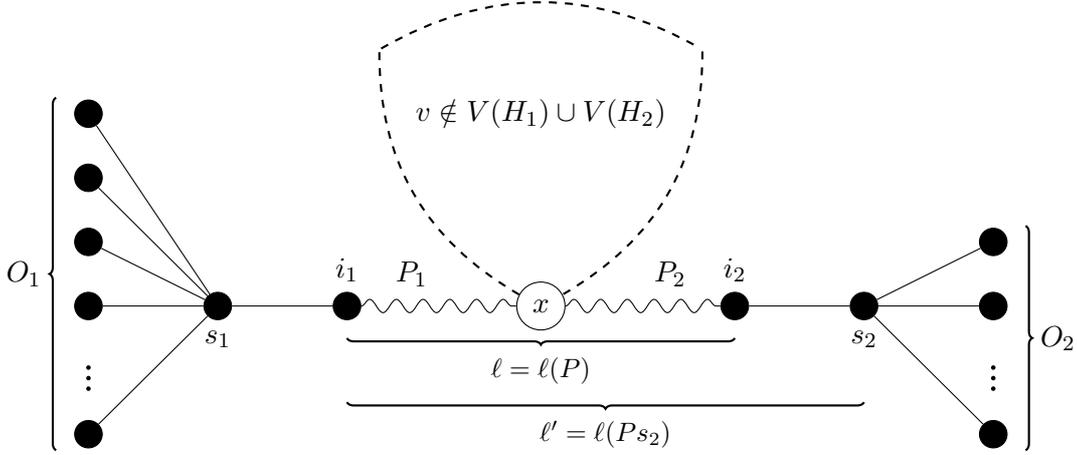

\begin{lemma} \label{lemma:phg-phg-hardness} 
Let $x$ be a cut vertex of an involution-free cactus graph~$H$. If there is a split of $H$ at $x$ into
$\{H_1,\dots,H_\kappa\}$, such that $(H_1,x)$ 
contains a partial hardness gadget $(s_1,i_1,O_1,P_1)$
and $(H_2,x)$ contains a partial hardness gadget $(s_2,i_2,O_2,P_2)$,
then $H$ contains a hardness gadget that satisfies the distance requirements
for every vertex $v\in V(H) \setminus V(P_1P_2s_2).$ 
\end{lemma}
\begin{proof} 
See Figure~\ref{fig:phg-phg}.
The two partial hardness gadgets ensure that $P=P_1P_2$
 is the unique shortest path in~$H$ from~$i_1$ to~$i_2$ and $P_1P_2s_2$
is the unique shortest path from~$i_1$ to~$s_2$.
Let $\ell = \ell(P)$ and $\ell' = \ell(Ps_2) = \ell+1$.

We first show that the number of $(\ell+2)$-walks from~$i_1$ to~$i_2$
(walks which use two more edges than $P$)
differs in parity from the number of $(\ell'+2)$-walks from~$i_1$ to~$s_2$.  To do this,
we use the sets of walks $T_1$, $T_2$ and $T_3$ from Definition~\ref{def:T-walks}.

\begin{figure}[t]
\begin{center}
\begin{tikzpicture}[scale=0.85]  
        \tikzstyle{vertex}=[fill=black, draw=black, circle]
	\tikzstyle{root}=[ draw=black, minimum size=1mm ,circle]

\node[vertex] (1) at (2,0) [label=-90:$x_{a'}$] {};
\node[vertex] (2) at (4,0) [label=-90:{$x_{\ell+1}=i_2$}] {};
\node[vertex] (3) at (6,0) [label=-90:$s_2$] {};
\node[vertex] (c1) at (4,2) {};
\node[vertex] (c2) at (6,2) [label=90:{$y=o_1$}] {};
\node[vertex] (c3) at (2,2) {};
\node[vertex] (4) at (8,0) [label=-90:$o_3$] {};
\node[vertex] (5) at (8,2) [label=-90:$o_2$] {};

\node[root] (0) at (0,0) {$x$};

\node (C) at (4,1) {$C$};

\draw (1)--(4);
\draw (3)--(5);
\draw (1)--(c3)--(c1)--(c2)--(3);
\draw[decorate, decoration=snake] (0)--(1);

\draw[decorate, decoration=brace, thick] (6,-1)--(2,-1);
\node at (4,-1.5) {\small{$r$}};
\end{tikzpicture}
\caption{An example of an impossible detour cycle $C$ in $T_1(Ps_2)$, using neighbours of $s_2$. Here $O_2=\{o_1,o_2,o_3\}$.\label{fig:contradiction-cycle}}
\end{center}
\end{figure}
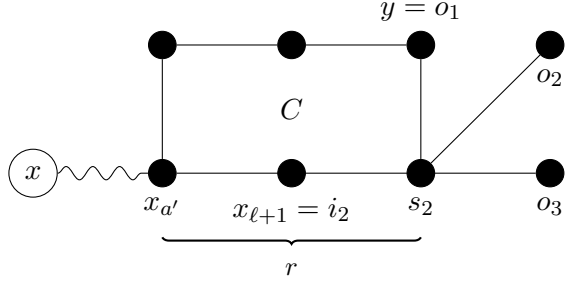

First, we show that $T_1(P) = T_1(Ps_2)$, by arguing that every detour
cycle~$C$ that is available from~$Ps_2$ is also available from~$P$.
Consider a walk $x_1\dots x_a P' x_b\dots x_{\ell'} s_2$ as in the
definition of $T_1(Ps_2)$, where $x_1=i_1$ and $x_{\ell'}=i_2$.  The
claim is obvious if $x_b\in P$ so suppose for contradiction that
$x_b=s_2$ and the detour cycle is $x_a\dots x_{\ell'}s_2P'x_a$.  Then,
by the definition of partial hardness gadgets, the neighbour~$y$
of~$s_2$ in $P'$ is in~$O_2$ (see
Figure~\ref{fig:contradiction-cycle}).  However, the definition also
requires that every vertex in~$O_2$ has a unique shortest path to~$x$
but the detour cycle gives two paths of length~$\ell(P_2)+2$ from $y$
to~$x$.

A similar argument shows that $T_2(P)=T_2(Ps_2)$.  This is obvious when
$x_{b'}\in P$. If $x_{b'}=s_2$, the same construction gives a
contradiction: the path from $y$ to~$x$ via $P''$ is shorter than the
one via~$s_2$, which the definition requires to be the unique shortest
$y$--$x$ path.

Finally, we show that $T_3(P)\not\equiv T_3(Ps_2)\imod 2$.  
This is because $s_2$~offers
an additional set of edges that may be repeated,
$\{(s_2,o) \mid o\in O_2\}$, and there are odd number of edges in this set since $|O_2|$ is odd.
 
We now construct the hardness gadget.
If the number of $(\ell+2)$-walks from~$i_1$ to~$i_2$ is even,
then let $t=i_2$ and $\beta=\ell+1$.
Otherwise, let $t=s_2$ and $\beta=\ell'+1$.
In either case, the number of $(1+\beta)$-walks from~$i_1$ to~$t$ is even.
Then let $s=s_1$, $O=O_1$, $i=i_1$ and $K=\emptyset$.

To see that this is a hardness gadget, consider $o\in O$ and $y\in O \cup \{i\}$.
$s$ is the unique vertex adjacent to~$o$ and to~$y$
which has an odd number of $\beta$-walks to~$t$.  (The odd number is one.)
The construction guarantees an even number of $(1+\beta)$-walks from~$i_1$ to~$t$.
 
 Now consider the primary distance requirement
$\dist_H(v,O\cup \{i\}) + \dist_H(v,t) > \beta-1$.
The unique shortest path in~$H$ from~$t$ to~$O\cup \{i\}$
has length $\beta-1$
so the requirement is satisfied for any vertex~$v$ that is not on this path.
This is guaranteed by the restriction on~$v$ in the statement of the lemma.
There are no secondary distance requirements, since $K=\emptyset$.
 \end{proof}

\begin{lemma}\label{lem:cycles-hardness}
 Let $(H,x)$ be a rooted cactus graph 
 containing a cycle $C=x_1x_2\dots x_\ell x_1$ 
 where $\ell\neq 4$.
 For
${j}\in[\ell]$, let $H_{j}$ be the component containing~$x_{j}$ in the graph $H- E(C)$. 
Suppose that the rooted graph $(H[(V(H)\setminus V(H_1))\cup\{x_1\}],x_1)$ is involution-free.
If $\ell$ is even, let $\calJ = [\ell]\setminus\{1,\ell/2+1\}$.
Otherwise, let $\calJ = [\ell]\setminus\{1\}$.
If, for each
$j\in  \calJ$,
$(H_j,x_j)$ is a mosaic, then $H$ contains a hardness gadget that satisfies
the distance requirements for each $v\in V(H_1)$. 
\end{lemma}
\begin{proof}
Note that for each $j\in \calJ$, $(H_j,x_j)$ is involution-free.
We start by dispensing with some easy cases.
First, if there is a
$j\in\calJ$ such that $(H_j,x_j)$ is a shortcut mosaic then, by
Corollary~\ref{corollary:mosaic-hardness}, $H_j$ contains a hardness gadget which satisfies the distance requirements
for $x_j$ and this is also a hardness gadget in~$H$ that satisfies the
distance requirements for all $v\in V(H)\setminus V(H_j)$, in
particular for all $v\in V(H_1)$.
Second,
suppose that there is a  $j\in\calJ$ such that $(H_j,x_j)$ is a  proper mosaic and 
$\deg_H(x_j)$ is
odd. $x_j$ is a cut vertex of~$H$ since
$H$ is a cactus graph. Therefore,   Lemma~\ref{lemma:mosaic-oddroot}
guarantees that $H$ has a hardness gadget which satisfies the distance requirements
for every $v\in V(H)\setminus V(H_j)$.
 
Thus, we can assume with loss of generality that
for every $j\in\calJ$, $(H_j,x_j)$ is a shortcut-free mosaic.
The two possibilities are:
\begin{itemize}
\item $(H_j,x_j)$ is a (possibly trivial) shortcut-free mosaic and $\deg_H(x_j)$ is even, or
\item $(H_j,x_j)$ consists of a single bristle.
\end{itemize}

Since $(H[(V(H)\setminus V(H_1))\cup\{x_1\}],x_1)$  is involution-free, there is some $j\in  \calJ$ 
such that $\deg_H(x_j)$ is even. Otherwise, for each 
$j\in   \calJ$,
 $(H_j,x_j)$ is a bristle. Hence $H$ has an
involution which fixes~$H_1$ but exchanges $H_{1+d}$ with 
$H_{\ell+1-d}$ for each 
$d\in \{1,\ldots, \lfloor\ell/2 \rfloor \}$.
We will consider two cases, depending on~$\ell$.
  
\medskip

\noindent\textbf{Case 1.} $\ell$ is odd.
We split the analysis into two cases.

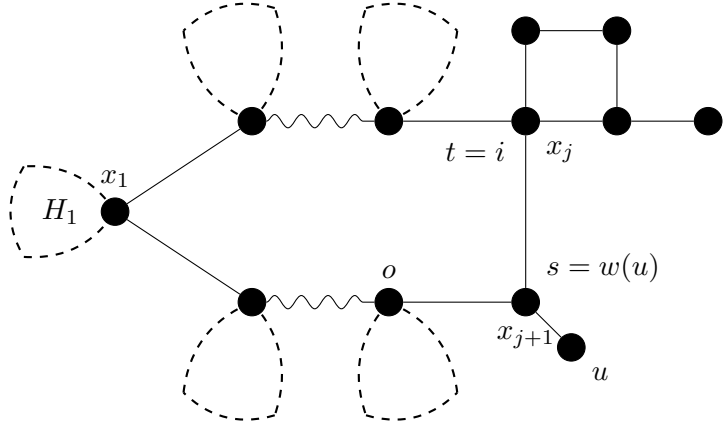
\begin{figure}[t]
\begin{center}
\begin{tikzpicture}[scale=.6]
        \tikzstyle{vertex}=[fill=black, draw=black, circle]

	\node[vertex] (1) at (0,0) [label=90:$x_1$] {};
	\node[vertex] (2) at (3,2)  {};
	\node[vertex] (3) at (6,2)  {};

	\node[vertex] (4) at (9,2)  [label=225:{$t=i$},label=-45:$x_j$] {};
	\node[vertex] (7) at (3,-2)  {};
	\node[vertex] (6) at (6,-2) [label=90:$o$] {};
	\node[vertex] (5) at (9,-2) [label=45:{$s=w(u)$},label=-90:$x_{j+1}$] {};
	\node[vertex] (5b) at (10,-3) [label=-45:$u$] {};

\node[vertex] (4a) at (9,4) {};
\node[vertex] (4b) at (11,4) {};
\node[vertex] (4c) at (11,2) {};
\node[vertex] (4d) at (13,2) {};
\draw (4d)--(4c)--(4b)--(4a)--(4)--(4c);

\draw[decorate, decoration=snake] (2)--(3);
\draw[decorate, decoration=snake] (7)--(6);

\draw (7)--(1)--(2);
\draw (3)--(4)--(5)--(6);
\draw (5)--(5b);

\path[dashed, thick] (0,0) edge [bend right] (-2,1)
              (-2,1) edge [bend right] (-2,-1)
	      (-2,-1) edge [bend right] (0,0);
\node (h1) at (-1.2,0) {$H_1$};
 
\path[dashed, thick] (3,2) edge [bend right] (3.5,4.5)
              (3.5,4.5) edge [bend right] (1.5,4)
	      (1.5,4) edge [bend right] (3,2);

\path[dashed, thick] (3,-2) edge [bend left] (3.5,-4.5)
              (3.5,-4.5) edge [bend left] (1.5,-4)
	      (1.5,-4) edge [bend left] (3,-2);

\path[dashed, thick] (6,2) edge [bend left] (5.5,4.5)
              (5.5,4.5) edge [bend left] (7.5,4)
	      (7.5,4) edge [bend left] (6,2);

\path[dashed, thick] (6,-2) edge [bend right] (5.5,-4.5)
              (5.5,-4.5) edge [bend right] (7.5,-4)
	      (7.5,-4) edge [bend right] (6,-2);
\end{tikzpicture}
\caption{The hardness gadget of Case~1.1 of Lemma~\ref{lem:cycles-hardness}. In this case, $K=\{u\}$.\label{fig:odd-opposite}}
\end{center}
\end{figure}

\noindent\textbf{Case 1.1.} There is a $j\in\{\ceil{\ell/2}, \ceil{\ell/2}+1\}$
such that $\deg_H(x_j)$ is even.
 See
Figure~\ref{fig:odd-opposite}.
Without loss of generality, suppose that $j= \lceil \ell/2\rceil$
(otherwise this could be achieved by relabelling the vertices of~$C$, going the other way around the cycle). 
We will construct a hardness gadget.
Let $\beta=1$, $s=x_{j+1}$, $t=x_{j}$, and $i=t$.
Let $o$ be the neighbour of~$s$ in~$C$ that is not~$t$.
(That is, $o=x_1$ if $\ell=3$ and $o=x_{j+1}$, otherwise.)
Let $O=\{o\}$ and let $K = \Gamma_H(s) \setminus \{o,i\}$.
For every $u\in K$ let $w(u)=s$ and $k(u)=\ell-1$. 

To see that this is a hardness gadget, consider $o$ and $y\in \{o,i\}$. $s$
is the only vertex that is adjacent to $o$, $y$ and~$t$.
However, there are an even number of vertices that are adjacent to $t$ since $\deg_H(t)$ is even.

Consider the $(\ell-1)$-walks from $w(u)=s$ to each $u\in K$,
noting that every vertex in~$K$ is a neighbour of~$s$.  Since
$s=x_{\ceil{\ell/2}+1}$, no $(\ell-1)$-walk from $s$ to one of its
neighbours can use any edge that is not in $H' = C \cup H_3 \cup \dots
H_\ell$.  Since each of $(H_3,x_3), \dots, (H_\ell,x_\ell)$~is a mosaic, $C$~is the only odd cycle
in~$H'\!$.  Since $\ell-1$ is even and the distance from $s$ to~$u$ is
one, which is odd, there is exactly one $(\ell-1)$-walk from~$s$ to
each of its two neighbours $i$ and~$o$ in~$C$ (going the long way around the
cycle) and there are no $(\ell-1)$-walks from $s$ to~$K$, the set of
its neighbours outside~$C$.

Now consider the primary distance requirement 
$\dist_H(v,\{o,i\}) + \dist_H(v,t) > 0$.  
This is satisfied for any $v\in V(H)\setminus\{t\}$, including each $v\in V(H_1)$. 
Finally, consider the secondary distance requirement
$\dist_H(v,s)+\dist_H(v,\{o,i,u\})>\ell-3$.
This is true for any $v\in V(H_1)$
since $\dist_H(x_1,s)+\dist_H(x_1,\{o,i,u\}) = \dist_H(x_1,s)
+ \dist_H(x_1,o) = \floor{\ell/2} + \floor{\ell/2} - 1 = \ell-2$.

\begin{figure}[t]
\begin{center}
\begin{tikzpicture}[scale=.7]
        \tikzstyle{vertex}=[fill=black, draw=black, circle]

	\node[vertex] (1) at (-2,0) [label=90:$x_1$] {};
	\node[vertex] (2) at (2,2) [label=45:$t$] {};
	\node[vertex] (3) at (4,2)  {};
	\node[vertex] (4) at (6,2)  {};
        \node[vertex] (5) at (8,2) [label=45:$i$] {};
	\node[vertex] (6) at (10,2)[label=-45:$s$] {};
	\node[vertex] (11) at (2,-2)  {};
	\node[vertex] (10) at (4,-2)  {};
	\node[vertex] (9) at (6,-2)  {};
        \node[vertex] (8) at (8,-2)  {};
	\node[vertex] (7) at (10,-2) [label=45:$o$] {};
	\node[vertex] (15) at (0,1) [label=-90:$x_2$] {};
	\node[vertex] (115) at (0,-1)  {};

\draw (1)--(15)--(2)--(3)--(4)--(5)--(6)--(7)--(8)--(9)--(10)--(11)--(115)--(1);

\foreach \x in {3,4,5} \node[vertex] (b\x) at (2*\x-2,4) {};
\foreach \x in {3,4,5} \path (b\x) edge (\x);
	
\node[vertex] (6a) at (11.4,3.6) [label=90:$u$] {};
\draw (6)--(6a);

\node[vertex] (7a) at (11.4,-3.6) {};
\draw (7)--(7a);

\node[vertex] (2a) at (2,4) {};
\node[vertex] (2b) at (0,4) {};
\node[vertex] (2c) at (0,2) {};
\node[vertex] (2d) at (3,5) {};
\draw (2d)--(2a)--(2b)--(2c)--(2)--(2a);

\draw[decorate, decoration=brace,thick] (8,1.5)--(2,1.5);
\draw node at (5,1) {$\beta-1$};

\path[dashed, thick] (1) edge [bend right] (-4,1)
              (-4,1) edge [bend right] (-4,-1)
	      (-4,-1) edge [bend right] (1);
\node (h1) at (-3.2,0) {$H_1$};

\foreach \x in {1,2,3,4} { \path[dashed, thick] (2*\x,-2) edge [bend right] (2*\x-.5,-4)
              (2*\x-.5,-4) edge [bend right] (2*\x+.5,-4)
	      (2*\x+.5,-4) edge [bend right] (2*\x,-2);   };

\path[dashed, thick] (115) edge [bend right] (-1.5,-2.5)
              (-1.5,-2.5) edge [bend right] (-.5,-3)
	      (-.5,-3) edge [bend right] (115);
\end{tikzpicture}
\caption{The hardness gadget of Case~1.2 of Lemma~\ref{lem:cycles-hardness}. Here, $\ell=13$, $\beta=4$, and
$K=\{u\}$.\label{fig:odd-close}}
\end{center}
\end{figure}
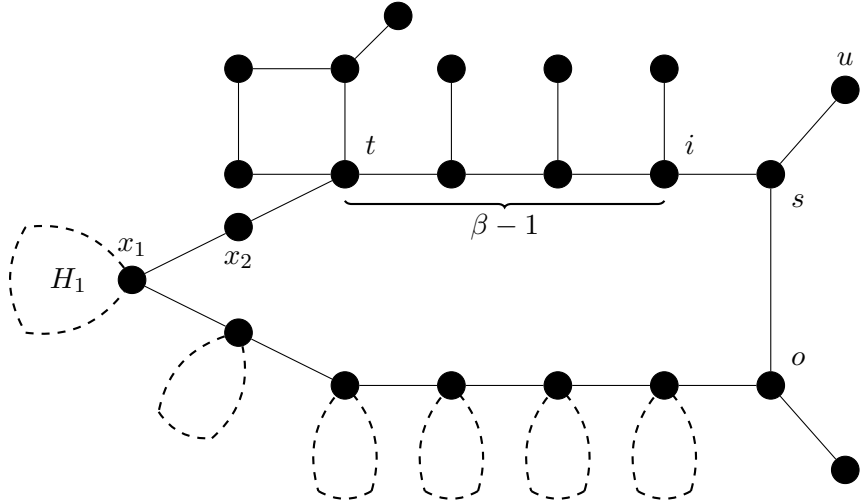

\noindent\textbf{Case 1.2.} We are not in Case~1.1 but there is still a $j\in\calJ$
such that $\deg_H(x_j)$ is even.  Without loss of generality, we may assume that
$j< \lceil \ell/2\rceil$, numbering the vertices of the cycle the
other way around, if necessary.

Again we will construct a hardness gadget.  See Figure~\ref{fig:odd-close}.
This time, let $s=x_{\lceil \ell/2\rceil}$
and $i=x_{\lceil \ell/2\rceil-1}$.  Since we are not in Case~1.1, $\deg_H(s)=3$.
Choose $t\in \{x_2,\ldots,x_{\lceil \ell/2\rceil-1} \}$
such that $\deg_H(t)$ is even and $\dist_H(t,i)$ is as small as possible.
Let $\beta=\dist_H(t,s)$.
Let $o=x_{\lceil \ell/2\rceil +1}$, which also has degree~3,
since we are not in Case~1.1; let $O=\{o\}$.
Let $u$ be the neighbour of~$s$ that is not in~$C$ and let $K=\{u\}$.
Let $w(u)=s$ and $k(u)=\ell-1$. 

To see that this is a hardness gadget, consider $o$ and $y\in \{o,i\}$.
$d_H(o,t)=\beta+1$ and $s$ is the only neighbour of~$o$ that has
any $\beta$-walks to $t$, and it has one such walk. Also, $s$ is adjacent to~$y$.
On the other hand,
there are an even number of $(\beta+1)$-walks from~$i$ to~$t$.
All of these walks repeat an edge, and there are an even number of edges that 
can be repeated because, by the choice of~$t$, every vertex between
$i$ and~$t$ has odd degree.
The proof that $w(u)=s$ has an even number of $k(u)$-walks to~$u$ and odd number of 
$k(u)$-walks to each of $o$ and~$i$ is exactly the same as in the proof of Case~1.1.

Now consider the primary distance requirement 
$\dist_H(v,\{o,i\}) + \dist_H(v,t) > \beta-1$.
This follows for any  $v$ which is not on the  unique shortest path in~$H$
from~$t$ to~$i$, which includes every $v\in V(H_1)$.
Finally, consider the secondary distance requirement
$\dist_H(v,s)+\dist_H(v,\{o,i,u\})>\ell-3$.
As in Case~1.1,
this is true for any $v\in V(H_1)$
since $\dist_H(x_1,s)+\dist_H(x_1,\{o,i,u\}) = \ell-2$.

\medskip

\noindent\textbf{Case 2.} $\ell$ is even.  Recall that $\ell\neq 4$
by the hypothesis of the lemma.
Choose $j\in\calJ$ such that $\deg_H(x_j)$ is even; again, we may
assume that $j\leq \ell-2$, and $1\notin \calJ$ by definition.
Let $s=x_{j+1}$. 
We will construct a hardness gadget for $H$ in each of two cases. 

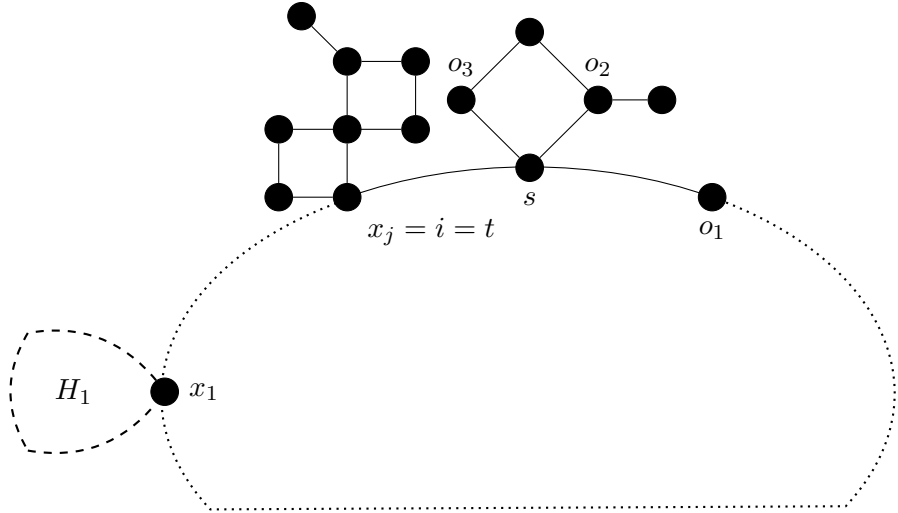
\begin{figure}[t]
\begin{center}
\begin{tikzpicture}[scale=.6]
        \tikzstyle{vertex}=[fill=black, draw=black, circle]
	
\draw (0,0) arc (60:120:8cm and 5cm);
\draw[dotted,thick] (0,0) arc (60:-30:8cm and 5cm) -- (-11,-6.9) arc (210:120:8cm and 5cm);

\node[vertex] (o) at (0,0) [label=-90:$o_1$] {};
\node[vertex] (i) at (-8,0) [label=-45:{$x_j=i=t$}] {};
\node[vertex] (s) at (-4,.65) [label=-90:$s$] {};

\node[vertex] (a) at (-12,-4.3) [label=0:$x_1$] {};
\path[dashed, thick] (a) edge [bend right] (-15,-3)
		    (-15,-3) edge [bend right] (-15,-5.6)
		    (-15,-5.6) edge [bend right] (a);
\node at (-14,-4.3) {$H_1$};

\node[vertex] (1) at (-5.5,2.15) [label=90:$o_3$] {};
\node[vertex] (2) at (-2.5,2.15) [label=90:$o_2$] {};
\node[vertex] (3) at (-4,3.65) {};
\node[vertex] (b1) at (-1.1,2.15) {};
\draw (b1)--(2)--(s)--(1)--(3)--(2);

\node[vertex] (5) at (-8,3) {};
\node[vertex] (6) at (-8,1.5) {};
\node[vertex] (7) at (-9.5,0) {};
\node[vertex] (8) at (-9.5,1.5) {};
\node[vertex] (9) at (-6.5,1.5) {};
\node[vertex] (10) at (-6.5,3) {};
\node[vertex] (b2) at (-9,4) {};
\draw (b2)--(5)--(6)--(i)--(7)--(8)--(9)--(10)--(5);
\end{tikzpicture}
\caption{An example of a hardness gadget for Case~2.1 of Lemma~\ref{lem:cycles-hardness}. Here $O=\{o_1,o_2,o_3\}$; the dotted line indicates omitted portions of the cycle~$C$.\label{fig:even-even}}
\end{center}
\end{figure}

\noindent\textbf{Case 2.1.} $\deg_H(s)$ is even.
See Figure~\ref{fig:even-even}.
Let $\beta=1$, $i=t=x_{j}$ and $O=\Gamma_H(s)\setminus \{i\}$. Note that $|O|$ is odd. 
Let $K=\emptyset$. 
 
To see that this is a hardness gadget, consider any $o\in O$ and
$y\in O\cup \{i\}$. $s$~is the unique vertex adjacent to $o$, $y$ and $t$.
Since $\deg_H(x_{j})$ is even, there are an even number of
$2$-walks from $i$ to $t=i$.
Consider the primary distance requirement
$\dist_H(v,O\cup \{i\}) + \dist_H(v,t) > 0$.
This is satisfied for every $v\neq t$, including every $v\in V(H_1)$. There are no secondary distance requirements since
$K=\emptyset$.

\begin{figure}[t]
\begin{center}
\begin{tikzpicture}[scale=.7]
        \tikzstyle{vertex}=[fill=black, draw=black, circle]

	\node[vertex] (1) at (0,0) [label=0:$x_1$] {};
	\node[vertex] (2) at (2,2) [label=45:$x_2$,label=-90:$w(x_2)$] {};
	\node[vertex] (3) at (4,2) [label=45:$x_3$,label=-90:$s$] {};
	\node[vertex] (4) at (6,2) [label=-90:$i$] {};
        \node[vertex] (5) at (8,0) [label=180:$x_5$] {};
	\node[vertex] (8) at (2,-2) [label=90:$t$] {};
	\node[vertex] (7) at (4,-2) {};
	\node[vertex] (6) at (6,-2) {};

\draw (1)--(2)--(3)--(4)--(5)--(6)--(7)--(8)--(1);

\node[vertex] (s) at (4,4) [label=0:$o$] {};
\draw (s)--(3);

\node[vertex] (a2) at (0,2) {};
\node[vertex] (b2) at (0,4) {};
\node[vertex] (c2) at (2,4) {};
\node[vertex] (2b) at (-1,3) {};
\draw (2b)--(a2)--(2)--(c2)--(b2)--(a2);

\path[thick, dashed] (1) edge [bend right] (-2,.5)
      (-2,.5) edge[bend right] (-2,-.5)
      (-2,-.5) edge [bend right] (1);
\node at (-1.2,0) {$H_1$};

\path[thick, dashed] (5) edge [bend right] (10,-.5)
      (10,-.5) edge[bend right] (10,.5)
      (10,.5) edge [bend right] (5);
\node at (9.2,0) {$H_5$};

\path[thick, dashed] (4) edge [bend left] (5.5,4)
      (5.5,4) edge[bend left] (6.5,4)
      (6.5,4) edge [bend left] (4);
\node at (6,3.2) {$H_4$};

\foreach \x in {6,7,8}{\path[thick, dashed] (\x) edge [bend left] (18.5-2*\x,-4)
      (18.5-2*\x,-4) edge[bend left] (17.5-2*\x,-4)
      (17.5-2*\x,-4) edge [bend left] (\x);
\node at (18-2*\x,-3.2) {$H_\x$};};
\end{tikzpicture}
\caption{An example of a hardness gadget for Case~2.2 of Lemma~\ref{lem:cycles-hardness}. In this case, $\ell=8$,
$\beta=3$ and $j=2$.\label{fig:even-odd}}
\end{center}
\end{figure}
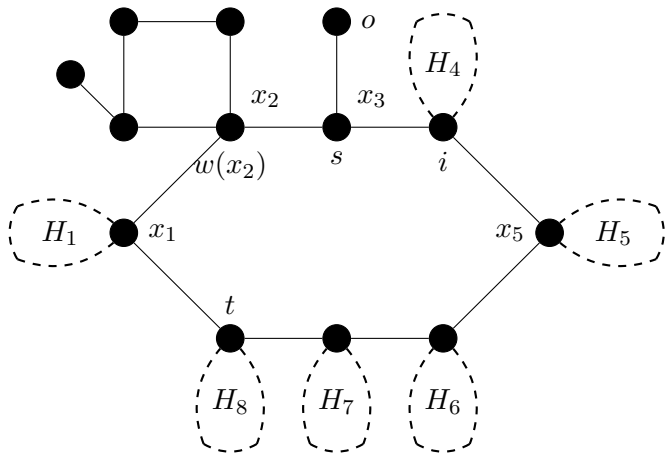

\noindent\textbf{Case 2.2.} $\deg_H(s)$ is odd so, in fact,
$\deg_H(s)=3$.  
See Figure~\ref{fig:even-odd}.

Let $\beta=\ell/2-1$. Let $i = x_{j+2}$, $K = \{x_{j}\}$ and $O =
\Gamma_H(s) \setminus \{x_j,x_{j+2}\}$, which is a single vertex: call
this~$o$.
Let $t$ be the unique vertex in~$C$ at distance $\ell/2$ from~$i$.
Let $w(x_{j})=x_{j}$ and $k(x_{j})=2$. 

To see that this is a hardness gadget, consider~$o$ and any $y\in \{o,i\}$.  
There is a unique $(\beta+1)$-walk from~$o$ to~$t$, and this goes
via~$s$.  Thus, $s$ is the unique vertex adjacent to $o$ and~$y$ that
has an odd number (one) of
$\beta$-walks to~$t$.
By construction, there are exactly two $(\beta+1)$-walks from~$i$
to~$t$, one going each way
around the cycle~$C$.
Since $\deg_H(x_{j})$ is even, there are an even number of $2$-walks from~$x_{j}$ to itself, but
there is exactly one $2$-walk from~$x_{j}$ to each of $o$ and~$i$.
The primary distance requirement is
$\dist_H(v,\{o,i\}) + \dist_H(v,t) > \beta-1$.
Since $\dist_H(t,\{o,i\})=\beta+1$, this holds for any~$v$.
Finally, the secondary distance requirement is
$\dist_H(v,x_{j})+\dist_H(v,\{o,i\}\cup K)>0$. 
This holds for any $v\neq x_{j}$, including all $v\in V(H_1)$.
\end{proof}


\section{Hardness gadgets in cactus graphs}

In this section, we show that every non-trivial involution-free cactus
graph contains a hardness gadget.  We first show that every
involution-free rooted cactus graph that is not a mosaic contains a
hardness gadget or a partial hardness gadget.

For rooted graphs, Observation~\ref{obs:rootsplit} allows an inductive
proof. Given an involution-free rooted graph~$(H,x)$ in which $x$~is a
cut vertex, we can take the split $\{H_1, \dots, H_\kappa\}$ and
recurse on the rooted graphs $\{(H_1,x), \dots, (H_\kappa, x)\}$, since
these are also involution-free.  If $x$~is not a cut vertex or
splitting at~$x$ does not give helpful rooted subgraphs, we instead
cut~$H$ up by deleting the edges of an appropriate cycle~$C$ to give
components which can be rooted at the vertices in~$C$.  One of these
contains~$x$, so needs special attention; the others are dealt with
inductively.

Finally, we need to show that every non-trivial, involution-free
unrooted cactus graph contains a hardness gadget.  To do this, we
temporarily introduce a root at a suitable vertex.

\begin{lemma}\label{lemma:phg-tree}
 An involution-free rooted tree $(H,x)$ with at least three vertices contains a partial hardness gadget.
\end{lemma}
\begin{proof}
Let $o$ be a leaf of~$H$ at maximal distance from~$x$.
Let~$s$ be the neighbour of~$o$.
Since $(H,x)$ is involution-free and $H$ contains at least three vertices
$\dist_H(x,o)>1$ so $s\neq x$.  Let~$i$ be the neighbour of~$s$ on the
path to~$x$.
Any neighbour of~$s$ that is not on the path to~$x$ must be a leaf as,
otherwise, there would be a leaf farther from~$x$ than $o$~is.  But no
vertex in an involution-free tree can be adjacent to more than one
leaf, so $\deg_H(s)=2$.
Let $P$ be the (unique) path in the tree~$H$ from~$i$ to~$x$.
The partial hardness gadget is $(s,i,\{o\},P)$.
\end{proof}

We say that a cut vertex~$x\in V(H)$ is \emph{cycle-separating} if at
least two of the components of the split of~$H$ at~$x$ contain cycles.

\begin{lemma}
\label{lem:lastday}
    Let $(H,x)$ be a connected, involution-free rooted cactus graph. Then at
    least one of the following is true:
    \begin{itemize}
    \item $H$ contains a hardness gadget satisfying the distance
        requirements for $x$, or
    \item $(H,x)$ contains a partial hardness gadget, or
    \item $(H,x)$ is a shortcut-free mosaic.
    \end{itemize}
\end{lemma}
\begin{proof}
The proof is by
induction on the number of cycles in~$H$.  If $H$ is acyclic and
is a single vertex or a single edge, $(H,x)$ is a shortcut-free
mosaic; if it is acyclic and has more than one edge, it contains a
partial hardness gadget by Lemma~\ref{lemma:phg-tree}.

Otherwise, $H$ contains at least one cycle.  If $x$~is a
cycle-separating cut vertex, let $\{H_1, \dots, H_\kappa\}$ be the split
of~$H$ at~$x$.  Every $(H_j, x)$ is an involution-free cactus graph
and has fewer cycles than~$H$.  If some $H_j$ contains a hardness
gadget that satisfies the distance requirements for~$x$, this is also
a hardness gadget in~$H$ and still satisfies the distance
requirements.  Likewise, a partial hardness gadget in some $(H_j,x)$
is also a partial hardness gadget in~$(H,x)$.  If there is no hardness
or partial hardness gadget in any $(H_j,x_j)$ then, by the inductive
hypothesis, every $(H_j,x)$ is a shortcut-free mosaic.  It follows
that $(H,x)$ is, itself, a shortcut-free mosaic.  It is a mosaic
because involution-freedom of~$(H,x)$ guarantees that the $(H_j,x)$
are pairwise non-isomorphic so, in particular, $x$~has at most one
bristle in~$H$.  It is shortcut-free because any shortcut in~$(H,x)$
must be inside some~$(H_j,x)$, but all of them are shortcut-free.

For the remainder of the proof, we assume that $x$ is not a
cycle-separating cut vertex (indeed, it is not necessarily even a cut
vertex).  Let $C=x_1 \dots x_\ell x_1$ be a cycle such that there is a
path from $x$ to~$x_1$ in which only $x_1$ is on a cycle.  (If $x$~is
on a cycle, then $C$~is this cycle, $x_1=x$ and the path is trivial.)
For $j\in [\ell]$, let $H_j$ be the component of $H-E(C)$ that
contains~$x_j$.  Thus, $x\in V(H_1)$.
We will use the fact below that $x=x_1$ if $x$ is on a cycle.
Otherwise, there is a unique path in~$H$ from~$x$ to~$x_1$.

Each of the rooted graphs $(H_j,x_j)$ is a cactus graph with fewer
cycles than~$H$ so, by the inductive hypothesis, each contains a
hardness gadget satisfying the distance restrictions for~$x_j$,
contains a partial hardness gadget, or is a shortcut-free mosaic.  If
any 
of $H_2,\ldots,H_\ell$
contains a hardness gadget, this is a hardness gadget in~$H$ so we
are done.

If $\ell$ is odd, let $\calJ = \{2, \dots, \ell\}$; otherwise, let
$\calJ = \{2, \dots, \ell\} \setminus \{\ell/2+1\}$.  Thus, $\calJ$ is
the set of indices~$j>1$ such that $H$~contains a unique shortest path
from $x$ to~$x_j$.

Suppose that, for some $j\in \calJ$, $(H_j,x_j)$ contains a partial
hardness gadget $(s,i,O,P)$.  In~$H$, $x$~has a unique shortest
path~$P'$ to~$x_{j}$ and $(s,i,O,P'P)$ is a partial hardness gadget
in~$(H,x)$.

Otherwise, for every $j\in\calJ$, $(H_j,x_j)$ is a shortcut-free
mosaic.  If $\ell\neq 4$ then, by Lemma~\ref{lem:cycles-hardness},
$H$~contains a hardness gadget that satisfies the distance
requirements for every vertex in~$H_1$, which includes~$x$.

We are left with the case $\ell=4$.  $(H_2,x_2)$ and $(H_4,x_4)$ are
mosaics and $(H_3, x_3)$ 
contains a partial hardness gadget or is a
shortcut-free mosaic.

\paragraph{Case 1.} $(H_3,x_3)$ contains a partial
hardness gadget.  

If $(H_2,x_2)$ is a proper mosaic then by Lemma~\ref{lemma:mosaic-deg23}, it contains a
2,3-path. Then, by Lemma~\ref{lem:phg-23path}, $H$ contains a hardness gadget 
that satisfies the distance requirements for every 
vertex $v\in V(H) \setminus(V(H_2)\cup V(H_3))$ and this includes $v=x$.
Similarly, there is a hardness gadget if $(H_4,x_4)$ is a proper mosaic.

So suppose that neither of $(H_2,x_2)$ and   $(H_4,x_4)$ is a proper mosaic.
Since $(H,x_1)$ is involution-free, one of $(H_2,x_2)$ and $(H_4,x_4)$ is a single edge and the other
is a single vertex. Suppose without loss of generality that $x_2$ is a single vertex. See Figure~\ref{fig:lastday}.
\begin{figure}[t]

\begin{center}
\begin{tikzpicture}[scale=.6]
        \tikzstyle{vertex}=[fill=black, draw=black, circle]
	\tikzstyle{root}=[ draw=black, minimum size=1mm ,circle]
	\node[root] (0) at (-4,0) {$x$};
	\node[vertex] (1) at (-1,0) [label=90:$x_1$] {};
	\node[vertex] (2) at (1,2)  [label=-90:$x_2$]{};
	\node[vertex] (3) at (3,0)  [label=90:$x_3$] {};
	\node[vertex] (4) at (1,-2) [label=90:$x_4$] {};
        \node[vertex] (4b) at (1,-4) {};

	\node[vertex] (i) at (6,0) [label=90:$i$] {};
	\node[vertex] (s) at (8,0) [label=90:$s$] {};

	\node[vertex] (o1) at (10,1) {};
	\node[vertex] (o2) at (10,0) {};
	\node[vertex] (o3) at (10,-3) {};
        \node (od) at (10,-1.5) {\huge{$\vdots$}};

\draw[decorate, decoration=brace, thick] (10.75,1.5)--(10.75,-3.5);
\node at (11.5,-1) {$O$};

\foreach \x in {o1,o2,o3} \draw (\x)--(s);
\draw[decorate,decoration=snake] (3)--(i);
      \node (p) at (4.5,-.5) {$P$};
\draw (s)--(i);

\draw (1)--(2)--(3)--(4)--(4b);
\draw (1)--(4);

\draw[decorate,decoration=snake] (0)--(1);

\path[dashed, thick] (-1,0) edge [bend right] (-5,2)
              (-5,2) edge [bend right] (-5,-2)
	      (-5,-2) edge [bend right] (-1,0);
\node (h1) at (-4.7,1.5) {$H_1$};
\end{tikzpicture}
 \caption{An example for Case~1 of Lemma~\ref{lem:lastday}.
 \label{fig:lastday}}
\end{center}
\end{figure}
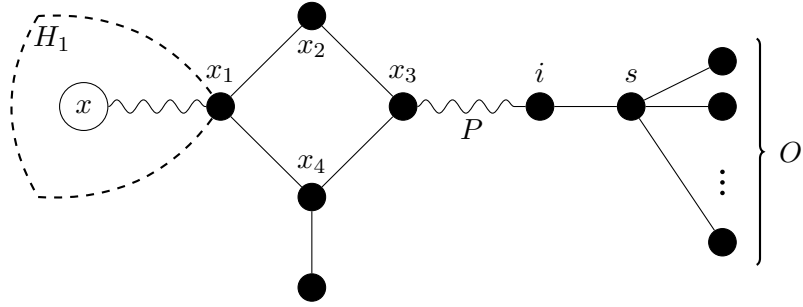
Now let $H'_1=C\cup H_1 \cup H_2 \cup H_4$ and let $P'$ be the empty path.
Then $(P',x_2,x_4)$ is a 2-3 path in $(H'_1,x_3)$. By Lemma~\ref{lem:phg-23path}, $H$ has
a hardness gadget that satisfies the distance requirements for
every $v\in V(H)\setminus (V(H_3)\cup \{x_2,x_4\})$ including $v=x$.

\paragraph{Case 2.} $(H_3,x_3)$ is a shortcut-free mosaic, which means
that $H' = C\cup H_2\cup H_3\cup H_4$ is an involution-free proper
mosaic when rooted at either $x_1$ or $x_3$.

First, suppose that $x\neq x_1$.  If $\deg_H(x_1)$ is odd then, by
Lemma~\ref{lemma:mosaic-oddroot}, $H$~has a hardness gadget that
satisfies the distance requirements for~$x_1$ and this also satisfies
the distance requirements for $x$.  If $\deg_H(x_1)$ is even, let
$s=x_1$, let $i$ be $s$'s~neighbour on the shortest path to~$x$ and
let $P$~be the unique shortest path from $x$ to~$i$ in~$H$.
$(s,i,\Gamma_{\!H}(s)\setminus \{i\},P)$ is a partial hardness gadget
in~$H$.

Finally, suppose that $x=x_1$.  Since $x$~is not a cycle-separating
cut vertex, every component of the split of~$H$ at~$x$ apart from
$(H'\!,x)$ is a tree.  If any of these contains more than one edge, it
contains a partial hardness gadget by Lemma~\ref{lemma:phg-tree}.
Otherwise, either $(H'\!,x)$ is the unique component of the split, or
there is exactly one other component, which is the one-edge tree
rooted at~$x$.  In either case, $(H,x)$~is a mosaic.  If $(H,x)$ is
shortcut-free, we are done; if not, it contains a hardness gadget
satisfying the distance requirements for~$x$, by
Corollary~\ref{corollary:mosaic-hardness}.
\end{proof}

We now show how to find hardness gadgets in unrooted cactus graphs by
choosing an appropriate vertex to act temporarily as a root.

\begin{theorem}
\label{thm:hardness-gadget}
Every involution-free cactus graph $H$ with more than one vertex contains a hardness gadget.
\end{theorem}
\begin{proof}
Split $H$ at a cut vertex~$x$ into rooted components $(H'_1,x), \dots,
(H'_k,x)$ with $|V(H'_1)|\geq \dots \geq |V(H'_k)|$,
choosing~$x$ to maximize $|V(H'_2)|$.  Each $(H'_j,x)$ is an
involution-free, rooted cactus graph and, if any of them
contains a hardness gadget satisfying the distance requirements
for~$x$, then this is also a hardness gadget in~$H$ and
we are done. Otherwise, by
Lemma~\ref{lem:lastday}
each $(H'_j,x)$ contains a partial hardness gadget or is a
shortcut-free mosaic.

If $|V(H'_2)|>2$, then each of $(H'_1,x)$ and $(H'_2,x)$ is either a
proper mosaic or contains a partial hardness gadget.
Therefore, $H$~contains a hardness gadget, by
Lemma~\ref{lemma:phg-phg-hardness} (two partial hardness gadgets),
Lemma~\ref{lemma:mosaic-mosaic-hardness}
(two proper mosaics) or Corollary~\ref{corollary:phg-mosaic-hardness}
(one of each).

$|V(H'_2)|$ cannot be~$1$ since $H$ is involution free.
So suppose that $|V(H'_2)|=2$.
$H$~is not a tree because then
there would be a vertex~$y$ with $\deg_H(y)\geq 3$, and
choosing $x=y$ would give $|V(H'_2)|>2$.  $H$~does not have two cycles because
every involution-free
cactus graph with two cycles
 contains a vertex~$y$ with $\deg(y)\geq 3$
and choosing $x=y$ would give $|V(H'_2)|>2$.  Further, $x$ must be on
$H$'s single cycle as, otherwise, we could have
chosen a vertex on the path from $x$ to the cycle as our cut vertex
and, again, obtained $|V(H'_2)|>2$.

Let the single cycle~$C$ of~$H$ be $x_1 x_2 \cdots x_\ell x_1$.
For $j\in [\ell]$, let $H_j$ be the component containing~$x_j$ in $H-E(C)$.
Clearly, $|V(H_j)| \leq 2$ --- otherwise we could have chosen $x=x_j$  to
achieve $|V(H'_2)| >2$.

For $H$ to be
involution-free, we must
have $\ell\geq 6$.
   Since $H$ is involution-free, the rooted graph
$(H[V(H)\setminus V(H_1)\cup \{x_1\}],x_1)$ is involution free, and
for each $j\in[\ell(C)]\setminus\{1\}$,
$H_j$ is an isolated vertex or a bristle, so
$(H_j,x_j)$
is a mosaic.
Lemma~\ref{lem:cycles-hardness} guarantees the existence of a hardness gadget
in $H$.
\end{proof}


\section{Homomorphisms to cactus graphs}

We now use hardness gadgets to show \parp{}-completeness of
\parhcol{H} for non-trivial involution-free cactus graphs~$H$.  This
is more complicated than the case for trees because an involution-free
cactus graph is not necessarily asymmetric --- recall, for example,
the graph~$H_4$ in Figure~\ref{fig:easy-hard}.  We first investigate
the automorphisms of cactus graphs and then give our reduction from
\paris{}.

\subsection{Automorphisms of cactus graphs}

A \emph{centre} of a graph~$H$ is a vertex~$x$ that minimises $\max_{y\in
  V(H)} \dist_H(x,y)$. 
It is well known that every tree has either one or two centres and that,
if a tree has two centres, they are connected by an edge. 
This fact was apparently first proved by Jordan in 1869~\cite{Jordan}.  

\begin{lemma}\label{lemma:vec-auto}
  Every cactus graph $H=(V,E)$ has a set $S \subseteq V$ such that $H[S]$ is (a) a vertex, (b) an edge, or (c) a cycle
and, for every $\pi \in \Aut(H)$, $\pi(S)=S$.
\end{lemma}
\begin{proof}
Let  $T$
be the tree formed from~$H$ by contracting all
edges of~$H$ that are on cycles.
Let $\rho\colon V(H) \to V(T)$ be the function that maps every vertex of~$V(H)$
to corresponding vertex in the contraction~$T$.

We will first consider the case that~$T$ has a single centre, $c$.
In this case, every automorphism of~$T$ fixes~$c$ so every automorphism
of~$H$ fixes $\rho^{-1}(c)$ setwise.
If $|\rho^{-1}(c)|=1$ then we are finished.
Otherwise, let
$C_1,\ldots,C_\kappa$ be the cycles of~$H$ 
containing vertices in~$\rho^{-1}(c)$.
Let $V(C_i)$ denote the set of vertices in~$C_i$ 
so $\rho^{-1}(c)=V(C_1)\cup \cdots \cup V(C_\kappa)$.
Let $T'$ be the graph with vertex set $\{C_1,\ldots,C_\kappa\}$
in which $(C_i,C_j)$ is an edge if and only cycles~$C_i$ and~$C_j$ intersect in~$H$.
Again, there are two cases.
First, suppose that $T'$ has a single centre, $C_i$.
In this case, every automorphism of~$T'$ fixes~$C_i$, so every automorphism
of $H[ \rho^{-1}(c)]$ fixes~$V(C_i)$ setwise. 
Thus, every automorphism of~$H$ fixes~$V(C_i)$ setwise so we are finished.
Second, suppose that $T'$ has two centres, $C_i$ and $C_j$.
In this case, every automorphism of~$T'$ fixes the edge $(C_i,C_j)$
so every automorphism of~$H[\rho^{-1}(c)]$ fixes $V(C_i)\cup V(C_j)$ setwise.
Since $H[\rho^{-1}(c)]$ is a cactus graph,  there is exactly one vertex, $w$,
in $V(C_i)\cap V(C_j)$. Thus, every automorphism of~$H[\rho^{-1}(c)]$ fixes~$w$
so every automorphism of~$H$ fixes~$w$ so we are finished.

Finally, suppose that $T$ has two centres~$c$ and~$c'\!$.
Every automorphism of~$T$ fixes the edge $(c,c')$.
Since $H$ is a cactus graph, there is exactly one edge $(u,u')$ of~$H$ with
$\rho(u)=c$ and $\rho(u')=c'\!$. Thus, every automorphism of~$H$ fixes the edge $(u,u')$
so it fixes $\{u,u'\}$ setwise so we are finished.
\end{proof}

\begin{lemma}\label{lemma:rcactus-involution}
If $(H,x)$ is a rooted cactus graph with a non-trivial automorphism, then it has an involution.
\end{lemma}

\begin{proof}
We prove this by induction on $n=|V(H)|$. 
The base case applies when $n\leq 2$. In this case, $(H,x)$ has no
non-trivial automorphisms.

For the inductive step, suppose $n>2$. 
If $\deg_H(x)=1$, then 
let $y$ be the neighbour of $x$ in $H$. 
If $(H,x)$ has a non-trivial automorphism then this induces a non-trivial automorphism
of $(H-x,y)$. By the inductive hypothesis, $(H-x,y)$ then has an involution,
which can be extended to an involution of $(H,x)$ by mapping~$x$ to itself.

So suppose that
$\deg_H(x)>1$ and that $(H,x)$ has a non-trivial automorphism.
Let $\{H_1,\ldots,H_\kappa\}$ be the split of $H$ at $x$. 
If there exist an $i\in[\kappa]$ such that $(H_i,x)$ has a
non-trivial automorphism, then 
by the inductive hypothesis, $(H_i,x)$ has an involution.
This  can be extended to an involution of $(H,x)$ by
fixing every vertex in $V(H) \setminus V(H_i)$.
Otherwise, there are distinct $i,j\in[\kappa]$ such that
$(H_i,x)$ and $(H_j,x)$ are
isomorphic. 
The automorphism that exchanges these and fixes all other vertices of~$H$ is
an involution of~$(H,x)$.
\end{proof}

\begin{lemma}\label{lemma:fixed-cycle}
Let $H$ be a cactus graph that has non-trivial automorphisms but no involution.
Then it
contains a cycle  $x_1 \dots x_{\ell} x_1$
such that every non-trivial automorphism of~$H$ induces a non-trivial
rotation of~$C$. Further, any two vertices at distance~2 in~$C$ are in different orbits under
the action of $\Aut(H)$.
  \end{lemma}

\begin{proof}
Note that $H$ is not a tree, since involution-free trees are asymmetric.
By Lemma~\ref{lemma:vec-auto},
there is a set $S\subseteq V[H]$ such that $H[S]$ is a vertex, an edge, or a cycle
and, for every $\pi \in \Aut(H)$, $\pi(S)=S$.

We start by showing that $H[S]$ is a cycle.
If instead $H[S]$ is a vertex~$x$ then $H$ has a non-trivial automorphism~$\pi$ which is
also an automorphism of~$(H,x)$. By Lemma~\ref{lemma:rcactus-involution} there is
an involution of $(H,x)$, contradicting involution-freedom of~$H$.
Similarly, if $H[S]$ is an edge~$(x,y)$ 
then $H$ has no automorphism that swaps~$x$ and~$y$ since this would give an involution.
This means that $H$ has a non-trivial automorphism that fixes~$x$. Once again, this 
is an automorphism of~$(H,x)$, which  contradicts the fact that~$H$ is involution-free.
So, let $H[S]$ be the cycle~$C=x_1  \dots x_{\ell} x_1$.
Note that the automorphism group of~$H$ fixes~$S=V(C)$ setwise
so the restriction of $\Aut(H)$ to~$S$ is a subgroup of the dihedral group acting on~$C$.
Since it does not contain a involution, it is a subgroup of the cyclic group
generated by the permutation~$g=(x_1 x_2 \dots x_{\ell})$.

So, for every non-trivial automorphism~$\pi$ of~$H$, the
restriction of~$\pi$ to~$V(C)$ is $g^{d_\pi}$ for 
some natural number~$d_\pi$.
To show that $g^{d_\pi}$ is a non-trivial rotation of~$C$,
we will show that $d_\pi\neq 0$.
For $i\in  [\ell]$, let $H_i$ be the connected component of $H-E(C)$ containing~$x_i$.
Suppose, for contradiction, that $d_\pi=0$.
Then there is an $i\in[\ell]$ such that the restriction of~$\pi$ to
$(H_i,x_i)$ is non-trivial. By Lemma~\ref{lemma:rcactus-involution}, $(H_i,x_i)$ has an involution~$\pi'\!$.
Therefore, $H$~has an involution, which agrees with $\pi'$ on $H_i$
and fixes every vertex outside~$H_i$.  This gives a contradiction.

Finally, suppose for contradiction  that $x_i$ and~$x_j$ are in the
same orbit, where $x_i,x_j\in C$ and $\dist_H(x_i,x_j)=2$.
Then there is an automorphism~$\pi$ of~$H$ with $\pi(x_i)=x_j$.
So $d_\pi=2$.  
But then $(H_{i'},x_{i'})$ and $(H_{j'},x_{j'})$ are isomorphic whenever $i'$ and~$j'$ have the same parity, 
so $H$ has
an involution which flips the cycle~$C$, contradicting the fact that 
$H$ is involution-free.
\end{proof}

\begin{definition}\label{def:homstar}
Say that a homomorphism $\pi$ from~$H$ to~$H$ (an endomorphism of~$H$)
is \emph{orbit-preserving} if $\pi(v)\in \Orb_H(v)$ for every $v\in
V(H)$.
Let $\Hom^\star(H,H)$ be the set of orbit preserving endomorphisms
of~$H$.
\end{definition}

The following lemma shows that the orbit-preserving endomorphisms of
an involution-free cactus graph are exactly its automorphisms.  Note
that this is not true for all cactus graphs.  For example, any even
cycle $x_1\dots x_{2\ell}x_1$ has an endomorphism mapping all
odd-numbered vertices to~$x_1$ and even-numbered vertices to~$x_2$.
This is orbit-preserving (any cycle has only one orbit) but not an
automorphism.

\begin{lemma}\label{lemma:cactus-pin-whole}
For any
  involution-free cactus graph $H$, $\Aut(H)=\Hom^\star(H,H)$.
\end{lemma}

\begin{proof}
It is immediate from  Definition~\ref{def:homstar} that $\Aut(H)\subseteq\Hom^\star(H,H)$. 

If $H$ is asymmetric, then the lemma follows from the fact that 
$\Aut(H)$
and $\Hom^\star(H,H)$ are both trivial.

So suppose that $H$ has a non-trivial automorphism.  Note that this
implies that $H$~is not a tree, since every involution-free tree is
asymmetric.
Assume for contradiction that there is a  $\phi\in\Hom^\star(H,H)$ 
that is not an automorphism of~$H$.
The homomorphism~$\phi$ cannot be a permutation  of $V(H)$, since a
bijective homomorphism is an isomorphism.
So there are vertices $x$, $y$, and~$z$ in $V(H)$ with 
$x\neq y$ and
$\phi(x)=\phi(y)=z$. 

Define~$C=x_1  \dots x_{\ell} x_1$ 
as in Lemma~\ref{lemma:fixed-cycle}.
For $i\in  [\ell]$,
let $H_i$ denote the connected component of $H-E(C)$ containing~$x_i$.
Each~$(H_i,x_i)$ is asymmetric since 
every non-trivial automorphism of~$H$ 
induces a non-trivial rotation of~$C$.
(If $(H_i,x_i)$ had a non-trivial automorphism then this could be extended to a non-trivial
automorphism of~$H$ that induces the trivial rotation of~$C$.)
Suppose that
$x\in V(H_i)$, $y\in V(H_{i'})$ and
$z\in V(H_{i''})$. Then $i'\neq i$ since
the definition of $\Hom^\star(H,H)$  
ensures that both~$x$ and~$y$ are in the orbit of~$z$ in $\Aut(H)$
but the only vertex in $V(H_i)$ that is in the orbit of~$x$ is $x$~itself.
  
There is an automorphism~$\pi$ of~$H$ with $\pi(x)=z$,
because $z\in \Orb_H(x)$.
$\pi$ induces a   rotation of~$C$,
$\pi(x_i)=x_{i''}$ so 
$\dist_H(x_i,x)=\dist_H(x_{i''},z)$.
Then, since $\phi(x)=z$ and $\phi$ is orbit-preserving and preserves edges of~$H$,
$\phi(x_i)=x_{i''}$.
Similarly, $\phi(x_{i'})=x_{i''}$.
Let $\pi'$ be an automorphism of~$H$ with $\pi'(x_{i''})=x_{i'}$
and consider the homomorphism 
$\phi' \in \Hom^\star(H,H)$ formed by applying~$\phi$ and then $\pi'\!$.
The homomorphism~$\phi'$ satisfies
$\phi'(x_i)=x_{i'}$ and $\phi'(x_{i'})=x_{i'}$.

Since $\phi'\in \Hom^\star(H,H)$ it maps every vertex of~$H$ to an element of its own orbit.
By Lemma~\ref{lemma:fixed-cycle} the orbit of every vertex in~$C$ is contained within~$C$. Thus, $\phi'(C)=C$.
To simplify the notation,
assume without loss of generality (by relabelling the vertices around~$C$ if necessary) that $i'=1$.
Then $1<i \leq\ell$ since $i\neq i'\!$.
Also, $\phi'(x_1)=x_1$.
Since $(x_1,x_2)$ is an edge, 
$\phi'(x_2)$ must be a neighbour of~$\phi'(x_1)=x_1$.
But $x_\ell$ is not in the orbit of~$x_2$ by Lemma~\ref{lemma:fixed-cycle} so
 $\phi'(x_2)=x_2$.
Similarly, $\phi'(x_j)=x_j$ for every $j\in  [\ell]$.
Specifically, $\phi'(x_i)=x_i$.
This contradicts $\phi'(x_i)=x_1$ since $i\neq 1$.
 \end{proof}

\subsection{Reduction from \paris{}} 
 
In the following definition, ``adding a new path~$P$ from $x$ to~$y$'' in a graph $G$ means forming a graph $G\cup P$
where $V(G)\cap V(P) = \{x,y\}$.

\begin{definition}
\label{def:G-Gamma}
Let $\Gamma=(\beta,s,t,O,i,K,k,w)$ be a hardness gadget in a graph~$H$ and let $G$ be any graph.  We
construct the graph $G_\Gamma$ as follows.  
Let $V''=\{v_e \mid e \in E(G)\}$ and
begin with the graph $G'=(V'\!,E')$
where
 $V' = V(G) \cup V(H) \cup V''$ (these three sets are assumed to be
disjoint) and $E'=E(H)$.
To~$G'\!$, add the following:
\begin{itemize}
\item for every vertex $x\in V(G)$, the edge $(x,s)$;
\item for every edge $e=(x,y)\in E(G)$, the edges $(x,v_e)$ and $(y,v_e)$;
\item for every edge $e\in E(G)$, a new $\beta$-path $P_{t,e}$ from $t$ to~$v_e$; and
\item for every vertex $x\in V(G)$ and every $u\in K$, a new $k(u)$-path $P_{x,u}$ from $x$ to~$w(u)$.
\end{itemize}
In~$G_\Gamma$,
we refer to vertices that are in~$V(G)$ as \emph{$G$-vertices} and
those in~$V(H)$ as \emph{$H$-vertices}. 
Figure~\ref{fig:completeness}
illustrates the construction.
\end{definition}

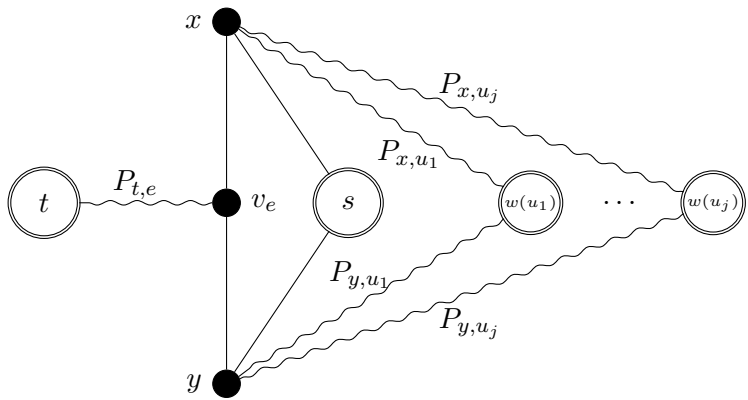
\begin{figure}[t]
\begin{center}
\begin{tikzpicture}[scale=.8]
        \tikzstyle{vertex}=[fill=black, draw=black, circle]
	\tikzstyle{pinned}=[draw=black,double, minimum size=1mm ,circle]

\node[pinned] (t)  at (-1,0) [inner sep=6.5pt] {$t$};
\node[vertex] (m)  at (2,0) [label=0:$v_{e}$] {};
\node[vertex] (x) at (2,3) [label=180:$x$] {};
\node[vertex] (y) at (2,-3)[label=180:$y$] {};
\node[pinned] (s)  at (4,0) [inner sep=6.5pt] {$s$};
\node[pinned] (w1) at (7,0) [inner sep=1pt] {\tiny{$w(u_1)$}};
\node (d) at (8.5,0) {$\dots$};
\node[pinned] (w2) at (10,0) [inner sep=1pt] {\tiny{$w(u_j)$}};

\node (px1) at (5,0.8) {$P_{x,u_1}$};
\node (px2) at (6,1.9) {$P_{x,u_j}$};
\node (py1) at (4.2,-1.2) {$P_{y,u_1}$};
\node (py2) at (6,-2) {$P_{y,u_j}$};

\node (ptxy) at (.5,0.3) {$P_{t,e}$};

\foreach \x in {w1,w2}{\foreach \y in {x,y} \draw[decorate, decoration={snake,amplitude=0.3mm}]
(\x)--(\y);};
\draw (x)--(s)--(y);
\draw (x)--(m)--(y);
\draw[decorate, decoration={snake,amplitude=0.3mm}] (m)--(t);
\end{tikzpicture}
\caption{The construction of $G_\Gamma$. Here $e=(x,y)\in E(G)$ so $x$ and $y$ are
$G$-vertices in $G_\Gamma$.  
The figure illustrates the induced subgraph of~$G_\Gamma$
corresponding to the edge $(x,y)\in E(G)$,
assuming $K=\{u_1,\ldots,u_j\}$. 
$H$-vertices are drawn with double circles; they will
be pinned in the proof of Theorem~\ref{thm:gadget-hard}.\label{fig:completeness}}
\end{center}
\end{figure}
 
Our construction of~$G_\Gamma$ is more complex than the one used  
for trees, because our hardness gadgets are more general than the
corresponding structures   in trees and because we must
deal with graphs~$H$ that are involution-free but still have
non-trivial automorphisms.  To see the problem of non-trivial
automorphisms, consider an involution-free cactus graph~$H$ that
contains a hardness gadget~$\Gamma\!$.  Suppose that $H$~has an
automorphism~$\pi$ that moves~$\Gamma\!$.  For our reduction, we wish
to pin one vertex to the $s$-vertex of~$\Gamma$ and another to the
$t$-vertex.  However, we cannot do this: we can only pin to the orbits
of these vertices, which include $\pi(s)$ and~$\pi(t)$, respectively.
We must avoid counting ``inconsistent'' homomorphisms that, for example,
map the first vertex to~$s$ and the second to~$\pi(t)$ because we do
not know how many such homomorphisms there are.  Including the copy
of~$H$ in~$G_\Gamma$ forbids such homomorphisms because, by pinning
each vertex of the copy of~$H$ to its own orbit, 
we only allow
homomorphisms whose restrictions to the copy of~$H$ are
orbit-preserving endomorphisms. By
Lemma~\ref{lemma:cactus-pin-whole}, these orbit-preserving endomorphisms
are automorphisms of~$H$.  
Hence, we
will count only those homomorphisms that map to $s$ and~$t$ and to $\pi(s)$
and~$\pi(t)$, and not any that map inconsistently between the copies
of $\Gamma\!$.

\begin{theorem}
\label{thm:gadget-hard}
$\parhcol{H}$ is $\parp$-complete for every involution-free cactus
graph~$H$ that contains a hardness gadget.
\end{theorem}
\begin{proof}  
Let $\Gamma=(\beta,s,t,O,i,K,k,w)$ be a hardness gadget in~$H$.  
Let $r=|V(H)|$.
We will give a polynomial-time Turing reduction from
$\paris(1, |O|)$ to $\rpinnedparhcol$.  
By the definition of hardness gadgets, $|O|$ is odd, so $\paris(1, |O|)$
is $\parp$-complete by Observation~\ref{obs:paris-complexity}.  
The result  will then follow from Theorem~\ref{thm:pinning}, which reduces
the problem \rpinnedparhcol{} to $\parhcol{H}$.

Let $G$ be an input to $\paris(1, |O|)$
and let $G_\Gamma$ be the graph defined in Definition~\ref{def:G-Gamma}.
Let $\pin$ be the $r$-restrictive pinning function $\pin \colon V(G_\Gamma) \to 2^{V(H)}$
with $\pin(v)=\Orb_H(v)$ for every $v\in V(H)$
and $\pin(w)=V(H)$ for every $w\in V(G_\Gamma) \setminus V(H)$.
The vertices that are restricted by the pinning function~$\pin$ are
drawn with double circles in Figure~\ref{fig:completeness}.
 We will establish the reduction 
 from $\paris(1, |O|)$ to $\rpinnedparhcol$
 by showing  that
 \begin{equation}
 \label{eq:desired}
  Z_{1,|O|}(G) \equiv {\left|\HomPin(G_\Gamma,H,\pin)\right|} \imod 2.
  \end{equation}

The pinning function~$\pin$ pins
every $H$-vertex in $G_\Gamma$   to its own orbit.
Therefore,   Lemma~\ref{lemma:cactus-pin-whole}  
shows that
the restriction of~$\phi$ to~$V(H)$ 
(which we denote $\phi|_{V(H)}$)
is an automorphism of $H$.
For $\pi \in \Aut(H)$, let 
$\Phi_\pi = \{ \phi \in \HomPin(G_\Gamma,H,\pin) \mid \phi|_{V(H)}= \pi\}$.
For any automorphisms~$\pi$ and~$\pi'$
of~$H$,
$|\Phi_\pi|=|\Phi_{\pi'}|$.
Thus, 
$${\left|\HomPin(G_\Gamma,H,\pin)\right|} \equiv
\left|\Aut(H)\right| \cdot \left| \Phi_\id\right|
\imod 2,$$
where $\id$ denotes the identity permutation. 
Since  $H$ is involution-free, 
Cauchy's Group Theorem guarantees that
$\left|\Aut(H)\right|$ is odd~\cite{McKay}.
Thus, 
$${\left|\HomPin(G_\Gamma,H,\pin)\right|} \equiv
  \left| \Phi_\id\right|
\imod 2.$$

For every $\phi\in \Phi_\id$,
$\phi(s)=s$ so, 
since $\Gamma_H(s) =  O \cup \{i\} \cup K$,
 every $G$-vertex~$v$ satisfies
$\phi(v)\in O\cup \{i\} \cup K$. 
Let $\Phi'_\id = \{\phi \in \Phi_\id \mid 
\forall v\in V(G), \phi(v)\in O\cup \{i\}
\}$. 
Consider the decomposition
$$\Phi_\id = \bigcup_{\rho: V(G) \rightarrow V(H)} \{ \phi\in \Phi_\id \mid \phi|_{V(G)}=\rho\}.$$
If $\rho(v)=u\in K$ for some $v\in V(G)$,
then 
$\left|\{ \phi\in \Phi_\id \mid \phi|_{V(G)}=\rho\}\right|$
is even.  This is because the definition of the hardness gadget means
there are an even number of $k(u)$-walks from $w(u)$ to~$u$ in~$H$ so
there are an even number of homomorphisms $\rho'\colon P_{v,w(u)}\to
H$ with $\rho'(v)=u$ and $\rho'(w(u))=w(u)$.
Thus, 
$\left| \Phi_\id \right| \equiv \left| \Phi'_\id \right| \imod 2$
so 
  $${\left|\HomPin(G_\Gamma,H,\pin)\right|} \equiv
  \left| \Phi'_\id\right|
\imod 2.$$

Finally, let $\Phi = \{ \phi \in \Phi'_\id \mid
\forall e=(x,y)\in E, \mbox{$\phi(x)\in O$ or
$\phi(y)\in O$}\}$. Once again, the definition of the hardness gadget ensures that
$| \Phi'_\id| \equiv |\Phi| \imod 2$
because, the even number of $(1+\beta)$-walks from $i$ to~$t$ in~$H$
guarantees an even number of homomorphisms with $\phi(x)=\phi(y)=i$.
So we have shown that 
${\left|\HomPin(G_\Gamma,H,\pin)\right|} \equiv
  \left| \Phi \right|
\imod 2$.
We will conclude the proof of~({\ref{eq:desired}}) by showing that $\left| \Phi \right| \equiv Z_{1,|O|}(G) \imod 2$.

Let $\Psi$ be the
set of functions $\pi \colon V(G) \to O\cup\{i\}$
such that $\pi^{-1}(i)$ is an independent set of~$G$.
Now 
$Z_{1,|O|}(G)  
\equiv |\Psi| \imod 2$ so it will suffice to show
$|\Psi| \equiv |\Phi| \imod 2$.
To do this, note that
$\Phi = \bigcup_{ \pi \in \Psi } \{ \phi \in \Phi_\id \mid \phi|_{V(G)}=\pi\}$.
So it suffices to show that, for every $\pi \in \Psi$,
$|\{ \phi \in \Phi_\id \mid \phi|_{V(G)}=\pi\}|$ is odd.
This follows from the definition of the hardness gadget.
Item~\ref{hgad2} in the definition ensures that for each edge~$e=(x,y)$,
there are an odd number of ways to extend~$\pi$
to $v_e$ and $P_{t,e}$.
Item~\ref{hgad4} ensures that there are an odd number of ways
to extend~$\pi$ to the internal vertices of the paths $P_{x,u}$
and $P_{y,u}$ for $u\in K$.
\end{proof}

We can now prove our main result.

{\renewcommand{\thetheorem}{\ref{thm:main}}
\begin{theorem}
Let $H$ be a simple graph in which every edge belongs to at most one cycle.
If the involution-free reduction of~$H$ has at 
most one vertex then $\parhcol{H}$ is solvable in polynomial time.
Otherwise, $\parhcol{H}$ is complete for $\parp$ with respect to polynomial-time Turing reductions.
\end{theorem}}
\begin{proof}
Let $H'$ be the involution-free reduction of~$H$.
If $H'$ has at most one vertex then $\parhcol{H'}$ is trivially solvable in polynomial time.
By Lemma~\ref{lem:same}, every graph~$G$ satisfies
 \begin{equation}
 \label{eq:usetwice}
 {\left| \Hom(G,H) \right|} \equiv {\left| \Hom(G,H') \right|} \imod 2
 \end{equation} 
so $\parhcol{H}$ is also solvable in polynomial time.

If $H'$ has more than one vertex then some component $H_1$ of~$H'$ has more than one vertex
(since $H'$ is involution-free). Also, $H_1$ is involution-free.
Since $H_1$ is an induced subgraph of~$H$,
it is a cactus graph. By Theorems~\ref{thm:hardness-gadget} and~\ref{thm:gadget-hard},
$\parhcol{H_1}$ is $\parp$-hard. By Lemma~\ref{lem:components},
$\parhcol{H'}$ is $\parp$-hard. But 
\eqref{eq:usetwice} gives a reduction from $\parhcol{H'}$ to $\parhcol{H}$, so $\parhcol{H}$ is also $\parp$-hard.
\end{proof}

\bibliographystyle{plain}
\bibliography{\jobname}

\end{document}